\documentclass[12pt, draftclsnofoot, onecolumn]{IEEEtran}
\usepackage{amsmath,epsfig,epsf,times,amssymb,booktabs,subfigure}
\usepackage{setspace}
\usepackage{array}
\usepackage{url}
\usepackage{float}
\usepackage{color}
\usepackage{graphicx}
\usepackage{epstopdf}
\usepackage{amsthm,cite}

\newtheorem{proposition}{Proposition}

\newtheorem{remark}{Remark}

\newtheorem{lemma}{Lemma}

\begin{document}

\title{Design and Analysis of Transmit Beamforming for Millimetre Wave Base Station Discovery}
\author{Chunshan Liu, Min Li, Iain B. Collings, Stephen V. Hanly, Philip Whiting
\thanks{Chunshan Liu, Iain B. Collings, Stephen V. Hanly and Philip Whiting are with the Department of Engineering, Macquarie University, Sydney, NSW 2109, Australia (email: \{chunshan.liu, iain.collings, stephen.hanly, philip.whiting\}@mq.edu.au.). Min Li was with Macquarie University and is now with the School of Electrical Engineering and Computer Science, The University of Newcastle, Callaghan, NSW 2308, Australia (email: min.li@newcastle.com.au). }
}

\maketitle

\begin{abstract}
In this paper, we develop an analytical framework for the initial access (a.k.a. Base Station (BS) discovery) in a millimeter-wave (mm-wave) communication system and propose an effective strategy for transmitting the Reference Signals (RSs) used for BS discovery. Specifically, by formulating the problem of BS discovery at User Equipments (UEs) as hypothesis tests, we derive a detector based on the Generalised Likelihood Ratio Test (GLRT) and characterise the statistical behaviour of the detector. The theoretical results obtained allow analysis of the impact of key system parameters on the performance of BS discovery, and show that RS transmission with narrow beams may not be helpful in improving the overall BS discovery performance due to the cost of spatial scanning. Using the method of large deviations, we 
identify the desirable beam pattern that minimises the average miss-discovery probability of UEs within a targeted detectable region. We then propose to transmit the RS with sequential scanning, using a pre-designed codebook with narrow and/or wide beams to approximate the desirable patterns. The proposed design allows flexible choices of the codebook sizes and the associated beam widths to better approximate the desirable patterns. Numerical results demonstrate the effectiveness of the proposed method.
\end{abstract}

\section{Introduction}
Millimetre-Wave~(mm-wave) communication is an important ingredient in the fifth generation cellular networks~(5G)~\cite{boccardi2014five,andrews2014will,7010531,niu2015survey} due to the large contiguous available bandwidth in the 30-300 GHz spectrum, and has already attracted much research attention~\cite{roh2014millimeter,6717211,6600706,el2013multimode}. Due to the high free-space pathloss at mm-wave frequencies, it is well understood that beamforming facilitated by a large number of antennas at Base Stations (BSs) and/or User Equipments (UEs) is essential to achieve reasonable coverage range and high spectrum efficiency~\cite{6834753}. The reliance on high-gain beamforming at the data transmission phase, however, has posed challenges in the design of initial access for mm-wave systems, as conventional designs for micro-wave (e.g., below 6~GHz) cellular systems may not work for mm-wave systems~\cite{li2013anchor}.

In the initial access of existing cellular systems (such as the Long Term Evolution (LTE)), BSs periodically broadcast reference signals~(RSs) in synchronisation channels and UEs detect the presence of the RS and determine the correct frame timer of the BS. Due to the favourable propagation characteristics in the sub-6 GHz bands, the RS in LTE is transmitted omnidirectionally from the BSs~\cite{sesia2009lte} without sophisticated beamforming. Such a transmission strategy, if adopted in mm-wave systems, could lead to a much smaller discoverable range than the intended coverage range~\cite{li2013anchor}. In other words, a UE that can achieve a reasonably high data-rate when high-gain beamforming is employed may not be able to discover and synchronise to the BS. In this case, communication link establishment fails at the first step and initial access will then be the bottleneck of a mm-wave cellular system. %This paper tackles the problem of mm-wave BS discovery and proposes an effective strategy for transmitting the RSs in the initial access.

{ The initial access for mm-wave cellular has been considered by~\cite{shokri2015millimeter,abu2014synchronization,7094805,barati2014dreictional}. Amongst them, reference~\cite{shokri2015millimeter} proposed a hybrid design that uses microwave band for initial access (BS discovery and synchronisation) and mm-wave band for subsequent beam alignment and data transmission. While this approach could provide macro-level discoverable range, one disadvantage is the increased hardware complexity due to the need for two radios (see Sec.~IV.B and Table-I of~\cite{shokri2015millimeter}).
\par References~\cite{abu2014synchronization,7094805,barati2014dreictional} instead focused on mm-wave standalone design and investigated the role of beamforming in mm-wave BS discovery.  Abu-Surra.~\emph{et al.} in~\cite{abu2014synchronization} numerically showed that a certain received signal to noise ratio~(SNR) of the RSs is required to guarantee a targeted probability of BS discovery. They then developed a beam-broadening technique to accommodate appropriate beamforming transmissions to attain the required SNR. Later, Desai.~\emph{et al.} in~\cite{7094805} proposed a hierarchical search strategy aided by the beamforming design from~\cite{6717211} to improve the SNR of the received RSs. If the transmission time for each beamformed RS is fixed, a narrower beam would provide higher received SNR and thus allow a beamformed UE to accumulate more energy for BS detection purpose, yielding a better BS discovery performance.

However, when beamforming is adopted, spatial scanning is inevitable to ensure reasonable discoverable range to UEs in all directions.
Consider multiple beams (with each pointing at a different direction) are multiplexed in time to scan the space and the time overhead incurred by spatial scanning is fixed. If narrower beams are deployed, higher beamforming gain can be achieved in each beamformed direction; but on the other hand, since more beams are required to cover all the directions of interest, a smaller fraction of time is made available to each beam. It hence remains unclear how the collected RS energy at UEs, and the performance of BS discovery, is affected by employing different beamforming strategies and different gain-time allocations. In this sense, it is natural to ask: What is the optimal RS transmission scheme that provides the best BS discovery performance for UEs within a targeted discoverable region?

The recent work in~\cite{barati2014dreictional} has compared omnidirectional (without beamforming) transmission with directional transmission of RSs via random beamformers. They derived generalised Likelihood Ratio Test (GLRT) detectors for BS discovery by assuming a rank-one (single-path) channel (see Sec. III of~\cite{barati2014dreictional}) and numerically evaluated the detection performance under both single-path and multi-path channel models. The results therein demonstrated that omnidirectional transmission outperforms directional transmission with random angular scanning. However, they did not consider more sophisticated RS transmission strategies such as sequential scanning with optimised beamformers to transmit of RSs.

\par In this paper, we propose a sequential beamforming strategy and develop a few fundamental limits on the performance of BS discovery. The analytical results in turn shed light on practical designs. Our contributions are three-fold stated as follows.
\par First, for the RS transmission strategy considered, we derive a GLRT detector for BS discovery that does not rely on a rank-one channel assumption. We explicitly characterise the statistical properties of the GLRT detector, from which we establish the relationship between the performance of mm-wave BS discovery and key system design parameters, including the RS sequence length and the time used for BS discovery at the UEs. These results provide useful guidance on proper choices of relevant parameters in practical system designs.
\par Second, by the method of large deviations~\cite{dembo2009large}, we show that it is the average beamforming gain that determines the asymptotic behaviour of the miss-discovery probability when searching time at UEs is large. We then identify desirable patterns of the average beamforming gain that minimises the average miss-discovery probability of UEs in an intended discoverable region.
\par Third, we propose a sequential beamforming strategy that uses a pre-designed beamforming codebook for RS transmission. This approach allows flexible choices between narrow and wide beams in constructing the codebook so as to better approximate the desirable patterns. The approach is shown to be beneficial when BS beamforming is constrained by per-antenna power limits. For practical choices of UE searching time, numerical results show that our proposed beamforming strategy provides orders of magnitude improvement compared to random beamformers.
}

\textbf{Notations}: boldface uppercase letters and boldface lowercase letters are used to denote matrices and vectors, respectively, e.g., $\mathbf{A}$ is a matrix and $\mathbf{a}$ is a vector. Notations $(\cdot)^T$ and $(\cdot)^\dag$ denote transpose and conjugate transpose, respectively. Notation $\|\mathbf{a}\|_2$ stands for the $l_2$ norm of vector $\mathbf{a}$, $||\mathbf{A}||_F$ stands for the Frobenius norm of matrix $\mathbf{A}$, and $\mathrm{Tr}\{\mathbf{A}\}$ stands for the trace of matrix $\mathbf{A}$. Finally, we use $\mathbb{E}\{\cdot\}$ to denote the expectation operation.

\section{System Model and Problem Description}\label{sec:Sig_model}
Consider a mm-wave communication system, in which each BS broadcasts its RS~(synchronisation sequence) periodically to allow UEs in its coverage range to detect the presence of the RS and to synchronise their frame timing. Fig.~\ref{fig:system_model} (a) depicts the frame structure of the transmit signal, where $T_{slot}$ is the duration of a slot and $T_{rs}<T_{slot}$ is the duration of each RS.

We also consider that directional transmissions of the RS are possible and the transmission pattern is repeated for every $J$ slots.
More specifically, as illustrated in Fig.~\ref{fig:system_model} (a), a set of $J$ beamformers  $\mathbf{w}_j \in {\mathbb C}^{N_T\times 1}$, $j=1,\ldots,J$, are sequentially employed for transmitting the RS sequence during $J$ consecutive slots. Here $N_T$ is the number of transmit antennas.

Denote the intended coverage angular space of a mm-wave BS as $\Omega$ and suppose that the beamformers come from a pre-designed codebook. To provide universal coverage of $\Omega$, a beamforming codebook consisting of multiple narrow beams or a codebook with a single but wide beam can be adopted, as illustrated in Fig.~\ref{fig:system_model}~(b) and (c). An omnidirectional transmission is a special case in such a setup and can be realised with $\mathbf{w}_j = [1,0,0,\ldots,0]^T$, $j = 1,\ldots,J$.

%To provide universal coverage of $\Omega$, a beamforming codebook consisting of multiple beamformers can be adopted, with narrow beams pointing in different directions, as illustrated in Fig.~\ref{fig:system_model}~(b). Alternatively, a codebook with a single beamformer of relatively wide beam may also be adopted, as illustrated in Fig.~\ref{fig:system_model}~(c).
%An omnidirectional transmission is thus a special case in such a setup and can be realised with $\mathbf{w}_j = [1,0,0,\ldots,0]^T$, $\forall j$.

\begin{figure}[t]
\centering
\includegraphics[width=0.45\textwidth]{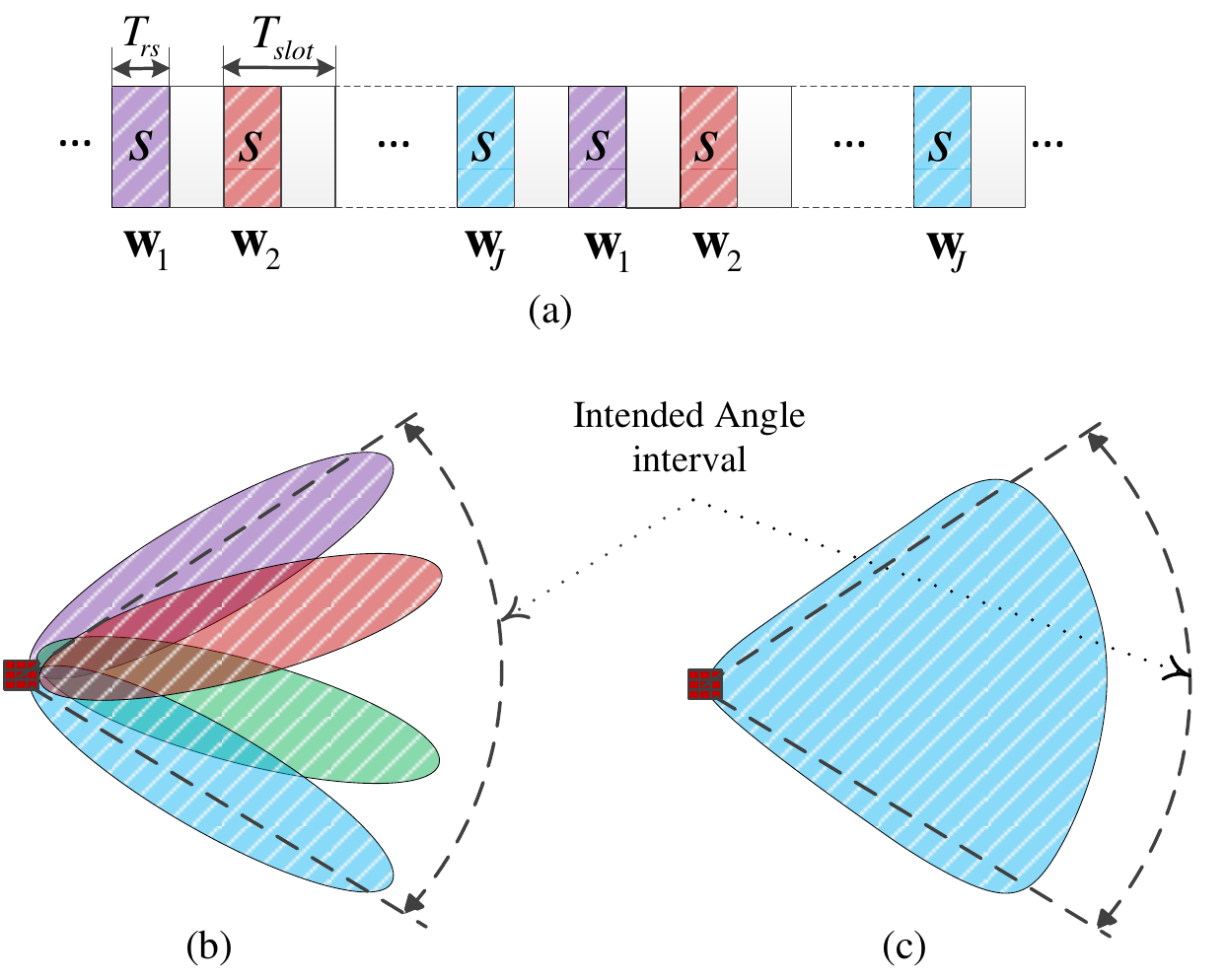}
\caption{An illustration of the system model: (a) frame structure; (b) $4$ beamformers are used to cover the intended angle interval; (c) a single beamformer is used to cover the intended angle interval.}
\label{fig:system_model}
\vspace{-0.5cm}
\end{figure}

The waveform of the RS in its complex baseband is denoted as $s(t)$, where $s(t) \neq  0$ iff $t \in [0,T_{rs}]$. The transmitted signal waveform $\mathbf{x}(t)$ is then represented as:
\begin{equation}\label{Tx_Sig_model}
\mathbf{x}(t) = \sum_{l=-\infty}^{+\infty}\mathbf{w}_{[l]_J}s\left(t-(l-1)T_{slot}\right),
\end{equation}
where $[l]_J$ is a modified modulo operation with $[l]_J = j$, if $l = kJ+j$, and $1\leq j \leq J$, $k \in \mathbb{Z}$.

The channel used to transmit the RS is sparse and is a superposition of a limited number of multi-path components, as evidenced by many measurements at mm-wave band (see~\cite{6834753} for an example). To facilitate the analysis, we assume that the channel is frequency flat and does not vary within each slot. Let $\mathbf{H}_l \in \mathbb{C}^{N_R \times N_T}$ denote the channel between the BS and a UE at slot~$l$, with $N_R$ being the number of receive antennas. Then, $\mathbf{H}_l$ can be represented as
\begin{equation}\label{eq:ch_multi_path}
\mathbf{H}_l = \sum_{q=1}^{Q}g_{l,q} \mathbf{u}^{\dagger}(\phi_{l,q})\mathbf{v}(\psi_{l,q}),
\end{equation}
where $Q$ is the number of multipath components, $g_{l,q}$ is the complex path coefficient of the $q$th path at the $l$th slot, and $\mathbf{u}(\phi_{l,q})$ and $\mathbf{v}(\psi_{l,q})$ are the steering vectors associated with Angle of Arrival (AoA) $\phi_{l,q}$ and Angle of Departure (AoD) $\psi_{l,q}$, respectively. In particular, considering a Uniformly Linear Array (ULA) with half wavelength antenna spacing at BS (or UE), steering vectors $\mathbf{u}(\phi_{l,q})$ (or $\mathbf{v}(\psi_{l,q})$) can be formulated as $[1,e^{j\pi\sin(\theta_{l,q})},\ldots,e^{j\pi(N_a-1)\sin(\theta_{l,q})}]$, with $N_a=N_T$ (or $N_a=N_R$) and $\theta_{l,q} = \phi_{l,q}$~(or $\psi_{l,q}$).
%\begin{equation}
%\mathbf{u}(\phi_{l,k}) = [1,e^{j\pi\sin(\phi_{l,k})},\ldots,e^{j\pi(N_T-1)\sin(\phi_{l,k})}].
%\end{equation}

With the channel defined, the received signal waveform at UE can be represented as:
{\small \begin{align}\label{Tx_Sig_model_2}
\mathbf{y}(t) & = \sum_{l=-\infty}^{+\infty}\mathbf{H}_{l}\mathbf{w}_{[l]_J}s\left(t-(l-1)T_{slot} - \tau_0\right)+\mathbf{z}(t)=\sum_{l=-\infty}^{\infty}\mathbf{h}_ls\left(t-(l-1)T_{slot} - \tau_0\right)+\mathbf{z}(t),
\end{align}}
where $\mathbf{h}_l \triangleq \mathbf{H}_l\mathbf{w}_{[l]_J}\in \mathbb{C}^{N_R \times 1}$ is the effective channel (after transmit beamforming) of slot $l$ and $\tau_0$ is a BS time offset and is unknown to UE initially. The noise $\mathbf{z}(t)$ is assumed to be spatially independent complex Gaussian with zero-mean and an unknown variance $\sigma^2$. In the following, we consider a fully digital receiver at UEs and note that power-efficient digital mm-wave receivers are possible by adopting low-resolution analog-to-digital convertors at UE, with a minimal loss of the received SNR~\cite{barati2014dreictional}. We also note that receiver combining using either analog combiner or hybrid combiner~\cite{alkhateeb2014channel} can be incorporated by \eqref{Tx_Sig_model_2} and thus our subsequent analysis applies. For instance, with a hybrid combiner, $\mathbf{h}_l \in \mathbb{C}^{N_{RF} \times 1}$ is the effective channel after transmit beamforming and receive combining, where $N_{RF}$ is the number of RF chains at the UE; in the case of adopting an analog combiner, the effective channel is a scalar ($N_{RF}=1$).

In the initial access, a UE attempts to detect the presence of $s(t)$ and find the correct synchronisation timer, i.e., $\tau_0$.
We assume that each UE performs detection based on signals collected from $L\geq J$ consecutive slots. Since the RS is broadcast in every slot,
it is sufficient to assume that $\tau_0$ lies in the interval $[0,T_{slot}]$. The UEs therefore only need to examine signals collected at time lags $\tau \in [0, T_{slot}]$ to determine if the RS is detected.

Without loss of generality, we consider that the signals from the observation interval $[\tau, \tau+LT_{slot}]$ are used for detection at time lag $\tau$. Since the RS has a duration $T_{rs}<T_{slot}$, only the signals from subintervals, $[(l-1)T_{slot} + \tau, (l-1) T_{slot} + T_{rs} + \tau]$, $l=1,\ldots,L$, are used to test if there is a reference signal and if the offset is $\tau$. We denote the signals from the $l$th subinterval~as 
\begin{equation}
\mathbf{y}_{l,\tau}(t')\triangleq\mathbf{y}(t+\tau)  = \mathbf{h}_ls(t-(l-1)T_{slot}+\tau-\tau_0) + \mathbf{z}(t),
\end{equation}
where $t' = t-((l-1)T_{slot}+\tau) \in [0,T_{rs}]$. Fig.~\ref{fig:received_model} illustrates the frame structure of the received signal slots used for detection. 

We further suppose that the UE samples the signals received in the observation period at rate $1/T_s$, with $N_{slot}$ samples per slot, where $T_s$ is the duration of a pilot symbol, and that $\tau_0$ is discrete and takes one of the $N_{slot}$ possible values. The discrete-time samples for $\mathbf{y}_{l,\tau}(t')$ are stored in matrix $\mathbf{Y}_{l,\tau} = [\mathbf{y}_{l,\tau}(T_s), \mathbf{y}_{l,\tau}(2T_s), \ldots, \mathbf{y}_{l,\tau}(N_sT_s)] \in \mathbb{C}^{N_R \times N_s}$, where $N_s\leq N_{slot}$ is the number of samples per slot for the RS. All the samples from the $L$ subintervals are finally organised as $\mathbf{Y}_{\tau}=[\mathbf{Y}_{1,\tau}^T,\ldots,\mathbf{Y}_{l,\tau}^T,\ldots,\mathbf{Y}_{L,\tau}^T]^T \in \mathbb{C}^{LN_R \times N_s}$.
The sampled RS is $\mathbf{s} = [s(T_s), s(2T_s), \ldots, s(N_sT_s)]^T\in \mathbb{C}^{N_s \times 1}$. Similar to~\cite{bliss2010temporal}, assuming the channel follows:
{\small\begin{equation}\label{channel_hypothesis}
\mathbf{h}_l(\tau)= \left \{ \begin{array}{ll}  \mathbf{h}_l,  & \tau = \tau_{0}\\
		 							0, & \tau \neq \tau_{0}\,
	\end{array}
\right.
\end{equation}}
the BS detection and synchronisation problem is formulated as the following binary test:
{\small \begin{equation}\label{Test_2}
\left \{ \begin{array}{ll} {\cal H}_1: &\mathbf{Y}_{\tau} =  \mathbf{h}\mathbf{s}^T + \mathbf{Z}_{\tau},\\
		{\cal H}_0: &  \mathbf{Y}_{\tau} = \mathbf{Z}_{\tau}.
	\end{array}
\right.
\end{equation}}
where $\mathbf{Z}_{\tau}\in \mathbb{C}^{LN_R\times N_s}$ follows a similar construction to $\mathbf{Y}_{\tau}$ and $\mathbf{h}=[\mathbf{h}_1^T,\ldots,\mathbf{h}_L^T]^T \in \mathbb{C}^{LN_R\times 1}$. The hypothesis test is repeated for all $N_{slot}$ discrete values of $\tau$.

\begin{figure}[t]
\centering
\includegraphics[width=0.6\textwidth]{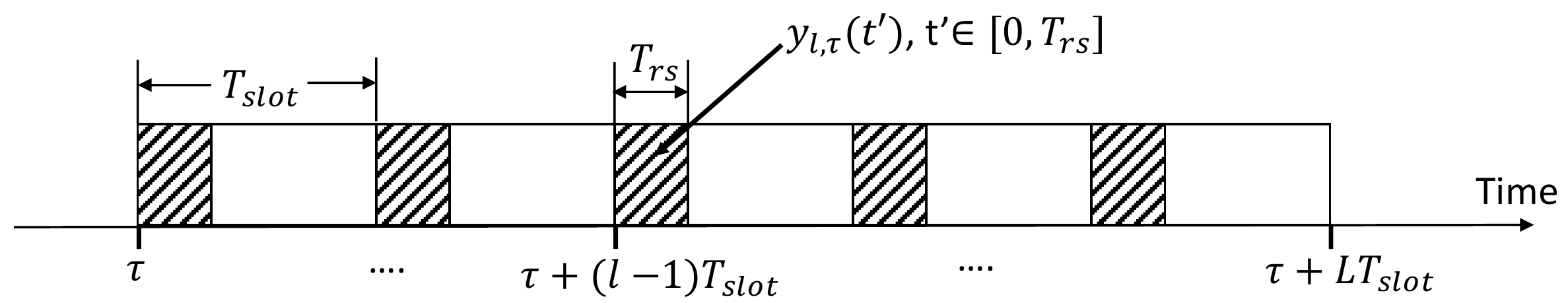}
\caption{Frame structure of the received signal used for detection.}
\label{fig:received_model}
\vspace{-0.6cm}
\end{figure}

\section{Generalised Likelihood Ratio Test for BS Discovery}\label{sec:GLRT}
In the initial access phase, it is reasonable to assume that UEs have no prior knowledge of any parameters in the test~\eqref{Test_2}, including the channel $\mathbf{h}$ and the noise variance $\sigma^2$. For this reason, we employ the GLRT method to perform the hypothesis testing.

The GLRT for this problem is represented as the following:
\begin{equation}\label{GLRT_statistic}
L'_G(\tau) = \frac{\max_{\mathbf{h},\sigma_1^2}p\left(\mathbf{Y}_{\tau}|{\cal H}_1;\mathbf{h},\sigma_1^2 \right)}{\max_{\sigma_0^2}p\left( \mathbf{Y}_{\tau}|{\cal H}_0; \sigma_0^2 \right)} \mathop{\gtrless}_{{\cal H}_0}^{{\cal H}_1} \gamma',
\end{equation}
where $L'_G(\tau)$ is the test statistic, $\gamma'$ is a threshold, $p\left(\mathbf{Y}_{\tau}|{\cal H}_1;\mathbf{h},\sigma_1^2\right)$ is the conditional probability density function~(PDF) of $\mathbf{Y}_{\tau}$ under ${\cal H}_1$ with given channel $\mathbf{h}$ and noise variance $\sigma^2_1$, and $p\left( \mathbf{Y}_{\tau}|{\cal H}_0; \sigma_0^2 \right)$ is the conditional PDF of $\mathbf{Y}_{\tau}$ under ${\cal H}_0$ for a given noise variance $\sigma^2_0$.

To derive the GLRT test statistic $L'_G(\tau)$, we first need to obtain the maximum likelihood~(ML) estimate of $\mathbf{h}$ and $\sigma_1^2$ under ${\cal H}_1$ and the ML estimate of $\sigma_0^2$ under ${\cal H}_0$. We refer derivations for the GLRT test statistic to Appendix~\ref{App_GLRT_detector} and present the results directly in the following proposition.

\begin{proposition}\label{Proposition:GLRT}
The test in \eqref{GLRT_statistic} is equivalent to the the following:
\begin{equation}\label{GLRT_sta_log}
L_G(\tau) = \frac{\sum_{l=1}^L \frac{1}{\|\mathbf{s}\|_2^2}\|\mathbf{Y}_{l,\tau}\mathbf{s}^*\|_2^2}{\sum_{l=1}^L \left(\|\mathbf{Y}_{l,\tau}\|_F^2 - \frac{1}{\|\mathbf{s}\|_2^2}\|\mathbf{Y}_{l,\tau}\mathbf{s}^*\|_2^2\right)} \mathop{\gtrless}_{{\cal H}_0}^{{\cal H}_1}\gamma,
\end{equation}
where $\gamma = \left(\gamma'\right)^{1/N}-1>0$ is the test threshold and $N=N_RN_SL$.
\end{proposition}

From~\eqref{GLRT_sta_log}, a UE attempting to discover a BS needs to compute the correlation between a candidate RS sequence, which can be viewed as the energy from the RS, and to calculate the total received energy  $\sum_{l=1}^L \|\mathbf{Y}_{l,\tau}\|_F^2$. Hypothesis ${\cal H}_1$ is accepted when the estimated energy from RS contributes to a significant portion of the total received energy, i.e., when $L_G(\tau) > \gamma$. It can also be seen from~\eqref{GLRT_sta_log} that the performance of the detector depends on system parameters including the length of RS sequence $N_s$ and the UE searching time $L$, and on the choice of the threshold $\gamma$. To analyse the impact of these parameters on the performance of mm-wave BS discovery, we present the following proposition.

%Based on the GLRT detector derived, a UE attempting to discover a BS needs to calculate the test statistic in~\eqref{GLRT_sta_log} and compare it with the threshold $\gamma$, which involves computing the correlation between a candidate RS sequence and the received signal samples and calculating the received signal energy.

%The performance of BS discovery based on the derived GLRT detector depends on system parameters including the length of RS sequence $N_s$ and the UE searching time $L$, and on the choice of the threshold $\gamma$. To analyse the impact of these parameters on the performance of mm-wave BS discovery, we present the following proposition.
\begin{proposition}\label{Proposition_GLRT_sta}
The GLRT test statistic $L_G(\tau) $ has the following statistical properties:
\begin{align} \label{F_distribution}
&(N_s-1)L_G(\tau)\sim \left \{ \begin{array}{ll}
				{\cal F}(2N_RL,2N_RL(N_s-1),0) & \text{under} ~~{\cal H}_0\\
				{\cal F}(2N_RL,2N_RL(N_s-1),\lambda) & \text{under}~~ {\cal H}_1
				\end{array}
			\right.
\end{align}
where
\[\lambda = \frac{2\|\mathbf{s}\|_2^2\sum_{l=1}^L\|\mathbf{h}_l\|_2^2}{\sigma^2} = \frac{2P_TN_s\sum_{l=1}^L\|\mathbf{h}_l\|_2^2}{\sigma^2}\]
and $P_T$ is the average transmit power. Here ${\cal F}(n_1,n_2,\lambda_1)$ denotes non-central F-distribution with degrees of freedom (DoFs) $n_1$ and $n_2$ and noncentrality parameter $\lambda$.
\end{proposition}

\begin{proof}
See Appendix~\ref{App_pGLRT}.
\end{proof}

Proposition~\ref{Proposition_GLRT_sta} characterises the distribution of the test statistic $L_G(\tau)$ using key system parameters such as the length of RS sequence $N_s$, the cell search interval $L$ and the transmission power. Since the PDF of the \emph{F}-distribution is known, it is convenient to evaluate the impact of these parameters on the performance of BS discovery. The results obtained can be used in choosing suitable values of these parameters to provide satisfactory BS discovery performance.

In particular, as the probability of false alarm,  i.e., the probability of declaring ${\cal H}_1$ under ${\cal H}_0$, is independent of the effective channels, it is convenient to choose $\gamma$ such that the probability of false alarm is below some target $P_{FA}$.
Recall that a UE needs to examine $N_{slot}$ time lags, corresponding to $N_{slot}$ tests.
 For the $i$th test with a hypothesised time lag $\tau_i$, the false alarm event is denoted as $FA_i = \{L_G(\tau_i) \geq \gamma~\text{under}~{\cal H}_0\}$. Then the probability of false alarm can be represented as $Pr\{\bigcup_{i=1}^{N_{slot}}FA_i\}$. Using the union bound, it can be shown that
\begin{align}
Pr\left\{\bigcup_{i=1}^{N_{slot}}FA_i\right\} &\leq \sum_{i=1}^{N_{slot}} Pr\{FA_i\} = N_{slot} Pr\{FA\}.
\end{align}
It then becomes clear that to meet a false alarm target $P_{FA}$, the test threshold $\gamma$ can be chosen such that the following equation is satisfied
\begin{equation}\label{Eq:threshold}
1-F(\gamma|2N_RL,2N_RL(N_s-1),0)= \frac{P_{FA}}{N_{slot}}
\end{equation}
where $Pr\{FA\} = 1-F(\gamma|2N_RL,2N_RL(N_s-1),0)$ and $F(x|n_1,n_2,\lambda_1)$ is the cumulative distribution function (CDF) of the \emph{F}-distribution ${\cal F}(n_1,n_2,\lambda_1)$.

We now focus on the hypothesis ${\cal H}_1$. Denote the probability of miss-detection (which is the probability of claiming ${\cal H}_0$ under ${\cal H}_1$) as $P_{miss} \triangleq Pr\{L_G(\tau)\leq\gamma|{\cal H}_1\}$. We note that under ${\cal H}_1$, both channel fluctuations and noise randomness can cause miss-detection, since the non-centrality parameter $\lambda$ depends on the effective channel gains. In the following, we denote $\bar{h}_L \triangleq \frac{1}{L}\sum_{l=1}^L\|\mathbf{h}_l\|^2_2$ as the sample mean of the effective channel gain and assume that for any $\xi>0$, there exists an $\underline{h}$ such that $Pr\{\bar{h}_L<\underline{h}\}=\xi$. Based on this assumption, we upper bound $P_{miss}$ in the following lemma that captures both channel fluctuations and noise randomness.

\begin{lemma}\label{lemma_fading}
Denote $\underline{\eta} = \frac{2P_TN_s}{\sigma^2}\underline{h}$. The probability of miss detection is upper bounded:
\begin{equation}\label{eq:upper_bound}
P_{miss}\leq \inf_{0<\xi<1 } \left\{\xi + (1-\xi)F(\gamma|2N_RL,2N_RL(N_s-1),L\underline{\eta})\right\}.
\end{equation}
\end{lemma}

\begin{proof}
See Appendix~\ref{proof_lemma}.
\end{proof}

Eq.~\eqref{eq:upper_bound} provides a convenient way to assess the probability of miss-detection under different system configurations and can thus provide guidance in choosing system parameters. Fig.~\ref{fig:upper_bound} provides an illustrative example, where the dashed lines are obtained from~\eqref{eq:upper_bound} and the solid lines are obtained from simulations.

In producing Fig.~\ref{fig:upper_bound}, the RS sequence length is fixed at $N_s=100$ and the number of samples in a slot is $N_{slot}=5000$, corresponding to a $2\%$ overhead due to the transmissions of RSs. The test threshold $\gamma$ is determined according to \eqref{Eq:threshold} with $P_{FA}=0.001$. The channels $\mathbf{H}_l$ are generated according to~\eqref{eq:ch_multi_path}, where there is one dominant path with time-invariant path-gain, AoA and AoD. The other $Q-1$ multipath components ($Q=6$) are assumed to have random AoAs and AoDs and that the path-gain coefficients are i.i.d. complex Gaussian. 
The ratio between the energy of the dominant path ($k=1$) and the total energy of the scattering paths ($q=2,\ldots,Q$) is set to $13.2$ dB, e.g., $10\log_{10}\frac{|g_{l,1}|^2}{\sum_{q=2}^Q\mathbb{E}\{|g_{l,q}|^2\}} = 13.2$, according to the measurement results reported in~\cite{muhi2010modelling} and adopted in~\cite{6600706}. Various values of the average SNR (of the RS signal) at each receive antenna, i.e., $SNR_{rs}=\frac{P_T}{N_R\sigma^2}\mathbb{E}\{\bar{h}_L\}$, are considered.
Eq.~\eqref{eq:upper_bound} is evaluated using the value of $\xi$ in the range $[10^{-5},1)$ and the corresponding $\underline{h}$ that yields the minimum value of the upper bound. The values of $\underline{h}$ are obtained numerically.

It can be seen that the theoretical results obtained capture well the behaviour of $P_{miss}$ under different system configurations. In particular, it provides satisfactory accuracy in choosing suitable system parameters to guarantee a certain miss-detection rate. For instance, the theoretical result for $SNR_{rs}= -23$ dB suggests $L=26$ for a target miss-detection rate of $10^{-3}$, which is very close to the value ($L=24$) obtained from simulations.

\begin{figure}[t]
\centering
\includegraphics[width=0.6\textwidth]{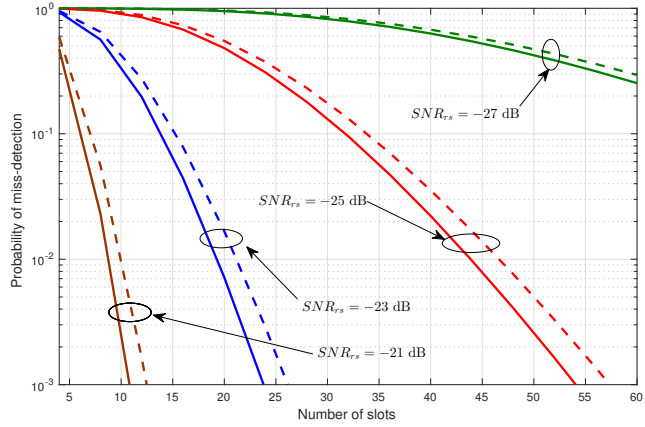}
\caption{Miss-detection probability versus the number of slots ($L$) used for BS discovery: the solid lines are obtained directly from simulation; the dashed lines are obtained by evaluating the upper bound given in \eqref{eq:upper_bound}.}
\label{fig:upper_bound}
\vspace{-1em}
\end{figure}

\begin{remark}
Conditioned on a sequence of channel realisations $\mathbf{H}_l$, $l=1,\ldots,L$, the probability of miss-detection is $F(\gamma|2N_RL,2N_RL(N_s-1),L\eta_L)$, where $\eta_L\triangleq \frac{\lambda}{L} = \frac{2P_TN_s\bar{h}_L}{\sigma^2}$ is the normalised non-centrality parameter depending on the transmit beamformers. Since the CDF of the \emph{F}-distribution, $F(n_1,n_2,\lambda)$ is monotonically decreasing with respect to $\lambda$ when the DoFs $(n_1,n_2)$ are fixed~\cite{baharev2008computation}, $F(\gamma|2N_RL,2N_RL(N_s-1),L\eta_L)$ is monotonically decreasing with respect to the mean channel gain $\bar{h}_L$ when $N_s$ and $L$ are fixed. To enhance the reliability of BS discovery for UEs in a certain direction, the RS transmit beamformers can be steered towards this direction to yield a larger $\bar{h}_L$, or equivalently, a large total energy received within the $L$ slots. However, this comes at the expense of sacrificing the beamforming gain in other directions, leading to poorer BS discovery performance in those directions. Clearly, if universal coverage for all directions is enforced, spatial scanning using narrow beams may not help to improve $\bar{h}_L$ as the total available time, $LN_s$, is allocated to a larger number of beams, each covering a different angular interval. We show next that a good choice is to design RS transmit beamformers such that $\eta_L$ for UEs at all directions are equalised. In the next section, we elaborate on this design.
\end{remark}

\section{Beamforming Design for RS Transmission}\label{sec:main_theorem}
In this section, we design the RS transmit beamformers for a group of UEs within a targeted detectable region and with certain target rates in a mm-wave band. We use the targeted data rates to calculate path loss parameters, which define the coverage topology and are later used to construct the RS transmit beamformers. Specifically, denote $R_{\text{th}}(\phi)$ as the target data rate for UEs at direction $\phi$. The corresponding SNR at data transmission phase, $\text{SNR}_{th}(\phi)$, can be determined according to $R_{\text{th}}(\phi) = \rho W \log_2\left(1+\text{SNR}_{th}(\phi)\right)$, where $\rho$ is the fraction of resource used for downlink data transmission (after excluding overhead) and $W$ is a typical bandwidth for each UE. Further, by assuming that $\text{SNR}_{th}(\phi)$ is achieved when the maximum transmit and receive beamforming gains (denoted as $G_T^{\max}$ and $G_R^{\max}$) are attained, an average path loss can be estimated as $\alpha(\phi) = \frac{P_T}{\sigma^2}\frac{G_T^{\max}G_R^{\max}}{\text{SNR}_{th}(\phi)}\frac{W_{rs}}{W}$, where $W_{rs}$ is the bandwidth used for RS transmission. To provide a universal BS discovery rate for these UEs, we propose to synthesise the RS transmit beamformers $\mathbf{w}_j$, $j=1,\ldots, J$ such that the average beamforming gain matches as close as possible to a desirable pattern:
\begin{equation}\label{Optimal_pattern_condition}
\frac{1}{J}\sum_{j=1}^{J}G_j(\phi) = G^*(\phi) = \left\{
				\begin{array}{ll}
				 \frac{\kappa^*}{\bar{\alpha}}\alpha(\phi), & \forall~\phi \in \Omega\\
				0. & \text{otherwise}
				\end{array}
		    \right.
\end{equation}
where $G_j(\phi)$ is the transmit beamforming gain of $\mathbf{w}_j$ in direction $\phi$ and $G^*(\phi)$ is the desirable beam pattern. Here $\bar{\alpha} \triangleq \frac{1}{|\Omega|}\int_{\Omega}\alpha(\phi)d\phi$  and $\kappa^* = \frac{4\pi}{|\Omega|}$ is a gain factor achieved when  beamformers $\mathbf{w}_j$, $j=1,\ldots,J$ all have zero energy leakage outside $\Omega$, i.e., $\frac{1}{4\pi}\int_{\Omega}G^*(\phi)d\phi =1$. Notation $|\Omega|$ denotes the solid angle spanned by angular interval $\Omega$.

We note that the target data rate $R_{\text{th}}(\phi)$ is defined as a function of direction $\phi$, and so is the path loss parameter $\alpha(\phi)$. The angular dependence is introduced to capture the possibility that the coverage of a mm-wave BS is circularly asymmetric due to uneven blockage across different directions, therefore UEs in some directions may be much closer to the BS and can achieve higher data rates than in other directions. In practice, this coverage information may not be immediately available at the BS.  So a uniform $R_{\text{th}}(\phi)$ over all directions may be adopted initially. By keeping the record of UEs' data rates at different directions (according to transmit beamforming directions), it is possible for the BS to acquire this coverage information. It is also noted that different choices of $R_{\text{th}}(\phi)$ can reflect different possible targeted UE groups. For instance, when $R_{\text{th}}(\phi)$ is set low, worst-case UEs at the coverage edge are considered; when $R_{\text{th}}(\phi)$ is set high, UEs closer to the BS (good UEs or moderate UEs) are considered.

In what follows, we first focus on a simplified, yet important, rank-one channel model, which captures well line-of-sight scenarios with a dominant path, and show that any set of beamformers that do not satisfy (14) provide only suboptimal miss-detection performance as $L$ goes large. We later present a systematic method to synthesise the beamformers to approach $G^*(\phi)$.

\subsection{Error exponent analysis of miss-detection probability as $L \rightarrow \infty$}
We consider the number of slots used for BS discovery is $L = KJ$ and the channels do not change during
the searching interval. Here $K \in {\mathbb N}$ is a scaling factor we used in the subsequent asymptotic analysis (e.g., when $K\rightarrow \infty$, it follows that $L \rightarrow \infty$). The normalised non-centrality parameter, $\eta(\phi)$, for a UE at direction $\phi$ can be written as:
\begin{equation}\label{eta_phi}
\eta(\phi) \triangleq\frac{\lambda(\phi)}{L} =\frac{1}{KJ} \frac{2P_TN_RN_s}{\alpha(\phi)\sigma^2}\sum_{k=1}^K\sum_{j=1}^{J}G_j(\phi)= \frac{2P_TN_RN_s}{\alpha(\phi)\sigma^2}G(\phi),
\end{equation}
where $G(\phi) = \frac{1}{J}\sum_{j=1}^{J}G_j(\phi)$ is the average beamforming gain in direction $\phi$. The probability of miss-detection for edge-UEs at direction $\phi$ can be represented as
\begin{equation}
p_{miss}(\phi)=F(\gamma|2N_RL,2N_RL(N_s-1),L\eta(\phi)).
\end{equation}

For assessing the overall performance of BS discovery, we consider the average miss-detection probability
\begin{equation}
\bar{p}_{miss} = \mathbb{E}_{\phi \in \Omega} \{p_{miss}(\phi)\} =\int_{\Omega} p_{miss}(\phi)p(\phi)d\phi
\end{equation}
as the main performance metric, where $p(\phi)$ is the PDF of the directions of the UEs being considered with $\int_{\Omega}p(\phi)d\phi=1$.

In the following, by applying the Gartner-Ellis theorem~\cite[Chapter 2.3, pp. 43]{dembo2009large}, we establish the large deviation principle of the test statistic $L_G$ when $\eta$ is fixed, and then characterise the asymptotic behaviour of $\bar{p}_{miss}$.
\begin{proposition}\label{Proposition:LDP}
The probability of miss-detection $p_{miss}=F(\gamma|2N_RL,2N_RL(N_s-1),L\eta)$ satisfies
\begin{align}\label{LDP:L}
\lim_{L\uparrow \infty} -\frac{1}{L}\log p_{miss} = \left\{\begin{array}{ll}
															I^*(\eta,\gamma), & \gamma<\frac{2N_R+\eta}{2N_R(N_s-1)} \\
															0, & \text{otherwise}
															\end{array}
															\right.
\end{align}
where $I^*(\eta,\gamma)$ is the rate function given by:
{\small
\begin{align}\label{rate_L}
I^*(\eta,\gamma) =& \frac{\eta}{2} \left(1-\frac{\gamma v^*}{N_R + \sqrt{N_R^2+\gamma \eta v^*}}\right) + N_R(N_s-1)\log \frac{2N_R(N_s-1)}{v^*} - N_R\log \frac{\gamma v^*}{N_R + \sqrt{N_R^2+\gamma \eta v^*}}.
\end{align}
}
Here $v^*= \frac{x^{*2}-N_R^2}{\eta\gamma}$ and $x^*>0$ is a solution to the following equation:
{\small
\begin{equation}
\frac{\gamma+1}{\eta\gamma}(x^2-N_R^2) - x - N_R -2N_R(N_s-1) = 0.
\end{equation}
}
Moreover, when $\gamma<\frac{2N_R+\eta}{2N_R(N_s-1)}$, $I^*(\eta,\gamma)$ is monotonically increasing in $\eta$.
\end{proposition}

\begin{proof}
See Appendix \ref{proof_LDP}.
\end{proof}

\begin{proposition}\label{proposition_optimal_beam}
Let $G(\phi)$ be the average beamforming gain of a set of $J$ beamformers used for RS transmission and let $\eta(\phi)$ be the normalised non-centrality parameter defined by \eqref{eta_phi}. If there exists
  $ \Omega^{-} \subset\Omega$ and $\eta^{-}$ such that $\eta(\phi)\leq \eta^{-}$ for all $\phi \in \Omega^{-}$ and $\int_{\Omega^-}p(\phi)d\phi>0$, the average miss-detection probability $\bar{p}_{miss}$ satisfies:
\begin{equation}
{\small \lim_{L\uparrow \infty}-\frac{1}{L}\log \bar{p}_{miss}\leq I^*(\eta^{-},\gamma), ~\text{if}~\gamma<\frac{2N_R+\eta^-}{2N_R(N_s-1)}}
\end{equation}
where $I^*(\eta,\gamma)$ is the rate function given by \eqref{rate_L}.
\end{proposition}
\begin{proof}
See Appendix~\ref{proof_proposition3}.
\end{proof}

Proposition~\ref{Proposition:LDP} demonstrates that an approximated upper bound to the miss-detection probability  can be obtained as follows:
\begin{equation}\label{Eq:approx}
{\small p_{miss}\approx\left\{ \begin{array}{ll}
                                e^{-LI^*(\eta,\gamma)}, & \gamma<\frac{2N_R+\eta}{2N_R(N_s-1)} \\
                                1,&\text{otherwise}.
                                \end{array}
                              \right.}
\end{equation}
Although the approximation provided in \eqref{Eq:approx} ignores all sub-exponential terms (thus the approximation does not coincide with the true $p_{miss}$ as $L$ increases), we have found that it provides satisfactory characterisation of the behaviour of $p_{miss}$, under various values of $\eta$, as shown in Fig.~\ref{fig:threshold_behaviour}. (In producing Fig.~\ref{fig:threshold_behaviour}, the value of $\eta$ is varied by considering different values of $SNR_{rs}$. Other system parameters are the same as those for Fig.~\ref{fig:upper_bound}.) In particular, for the range of $L$ for which $p_{miss}$ is small (e.g., $p_{miss}<10^{-2}$), it can be seen that \eqref{Eq:approx} provides an accurate approximation to the slope of $p_{miss}$ as $L$ increases. Moreover, a larger $\eta$ leads to a steeper slope, which is consistent with the monotonicity of $I^*(\eta,\gamma)$ with respect to $\eta$.

\begin{figure}[t]
\centering
\includegraphics[width=0.6\textwidth]{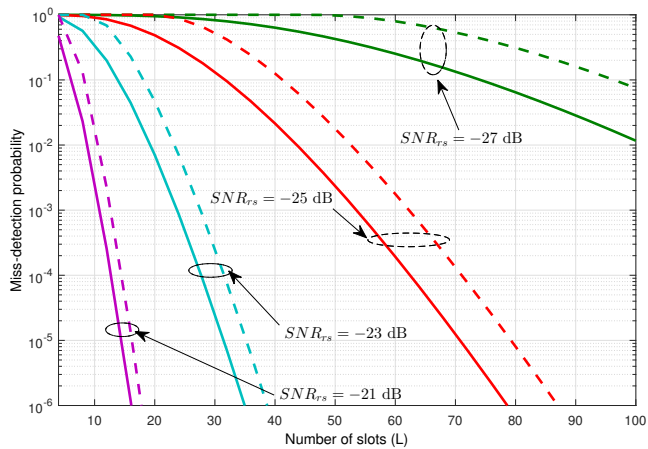}
\caption{Miss-detection probability versus the number of slots ($L$) used for BS discovery: the dashed lines are obtained by applying the approximation in \eqref{Eq:approx}.}
\label{fig:threshold_behaviour}
\vspace{-0.6 cm}
\end{figure}

Proposition~\ref{proposition_optimal_beam} further demonstrates that the asymptotic behaviour of $\bar{p}_{miss}$ is dominated by those directions with the smallest normalised non-centrality parameter. RS transmit beamformers that maximises the minimum possbile $\eta^{-}$ yield the steepest slope, and thus will outperform beamformers with a smaller $\eta^{-}$ when $L$ is sufficiently large. This insight motivates designing RS transmit beamformers by solving the following max-min optimisation problem:
\begin{align}\label{max_min}
\max_{G(\cdot)}\min_{\phi \in \Omega} \eta(\phi),~\text{s.t.}~\frac{1}{4\pi}\int_{\Omega}G(\phi)d\phi \leq 1,
\end{align}
where the constraint follows from the law of energy conservation.
When $N_R$, $N_s$ and $P_T$ are fixed, \eqref{max_min} is equivalent to:
\begin{align}\label{max_min_G}
\max_{G(\cdot)}\min_{\phi \in \Omega} \frac{G(\phi)}{\alpha(\phi)},~\text{s.t.}~\frac{1}{4\pi}\int_{\Omega}G(\phi)d\phi \leq 1
\end{align}
It is easy to check that $G^*(\phi)$ given by \eqref{Optimal_pattern_condition} is a solution to \eqref{max_min_G}.

\begin{remark}
The zero-leakage outside the angular interval is imposed when $\Omega$ is a subinterval of the entire angular space, i.e., $|\Omega|<4\pi$. When this constraint is imposed, the desired pattern \eqref{Optimal_pattern_condition} may not be achieved exactly with a finite number of transmit antennas. %Beamformers that approximate the pattern \eqref{Optimal_pattern_condition} are thus expected in practical implementations. 
Moreover, the requirement of matching the beamforming gain to $\alpha(\phi)$ poses an additional challenge in getting a good approximation to \eqref{Optimal_pattern_condition} when $\alpha(\phi)$ is uneven across $\Omega$. Driven by these observations, in the following, we propose a beamformer synthesis strategy that uses $M\leq J$ beamformers to approximate \eqref{Optimal_pattern_condition}. The proposed method is to select between several candidate values of $M$ the one that provides the best approximation to \eqref{Optimal_pattern_condition}.
\end{remark}

\subsection{Beamforming strategy for RS transmission}
We propose to approximate \eqref{Optimal_pattern_condition} based on a beamforming codebook of size $M\leq \min\{N_T,J\}$. The $M$ beamformers are used in sequence to transmit the RS during the $J$ consecutive slots. Since $M\leq J$, some of the $M$ beamformers could be used in more than one slot.
We use the following three-step procedure to construct the beamformers and perform slot allocation:
\begin{enumerate}
\item The whole intended angular space is partitioned into $M$ non-overlapping subintervals: $\Omega^{(m)}$, $m=1,\ldots,M$, $\bigcup_{m=1}^M\Omega^{(m)} = \Omega$ and $\Omega^{(m_1)}\cap\Omega^{(m_2)}=\emptyset$, $\forall m_1\neq m_2$.
\item Beamformer $\mathbf{w}^{(m)}$, $m=1,\ldots,M$ is synthesised to approximate the following pattern:
\begin{equation}\label{Optimal_pattern}
\small{G^{(m)}(\phi) = \left\{
				\begin{array}{ll}
				\kappa^{(m)} \frac{\alpha(\phi)}{\bar{\alpha}^{(m)}}, & \forall~\phi \in \Omega^{(m)}\\
				0, & \text{otherwise}
				\end{array}
		    \right.}
\end{equation}
where $\kappa^{(m)} = \frac{4\pi}{|\Omega^{(m)}|}$ and $\bar{\alpha}^{(m)} = \frac{1}{|\Omega^{(m)}|}\int_{\Omega^{(m)}}\alpha(\phi)d\phi$.
\item Beamformers $\mathbf{w}^{(m)}$s are used in sequence to transmit the RS, i.e., during the $J$ consecutive slots, with $J_m$ slots using beamformer  $\mathbf{w}^{(m)}$, where
\begin{equation}\label{Eq:number_beamformer}
\small{J_m = \frac{J|\Omega|^{(m)}\bar{\alpha}^{(m)}}{M|\Omega|\bar{\alpha}}.}
\end{equation}
\end{enumerate}
\vspace{-1.2em}
As an example, let us consider the scenario presented in Fig.~\ref{fig:system_model} (b), where the $60^\circ$ sector is partitioned into $M=4$ sub-intervals, i.e., $|\Omega^{(m)}| = |\Omega|/4$, $m=1,\ldots,4$. Further, assuming $\alpha(\phi) = \bar{\alpha}$, $\forall \phi \in \Omega$, the desired beamformer pattern in \eqref{Optimal_pattern} is uniform within each sub-interval with zero-leakage outside the interval, as illustrated in Fig.~\ref{fig:construct_beam_example}. Using \eqref{Eq:number_beamformer}, it can be obtained that $J_m = \frac{J}{4}$, $m=1,\ldots,4$. This implies that one can sequentially transmit the RS using the $M=4$ beamformers for $4$ consecutive slots, and repeat for $J/4$ times to meet \eqref{Optimal_pattern_condition} (see Fig.~\ref{fig:construct_beam_example} (b)).

 \begin{figure}[t]
\centering
\begin{subfigure}%{0.5\textwidth}
  \centering
  \includegraphics[width=0.3\linewidth]{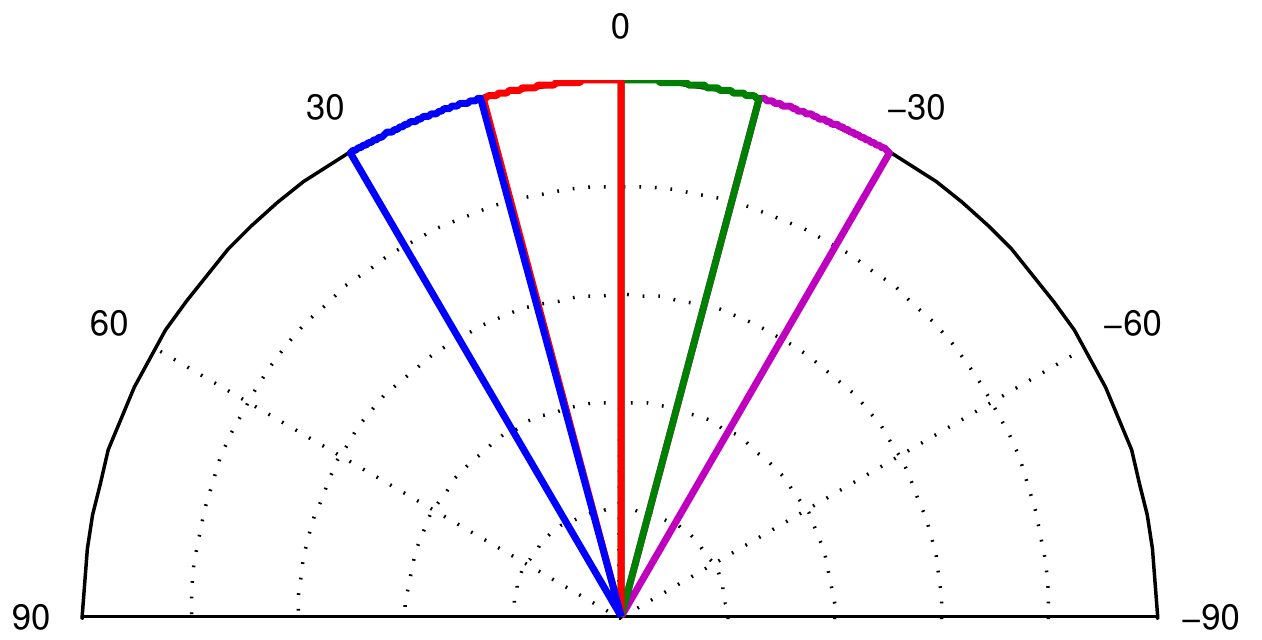}
  \\
  (a)
\end{subfigure}
\\
\begin{subfigure}%{0.5\textwidth}
  \centering
  \includegraphics[width=0.45\linewidth]{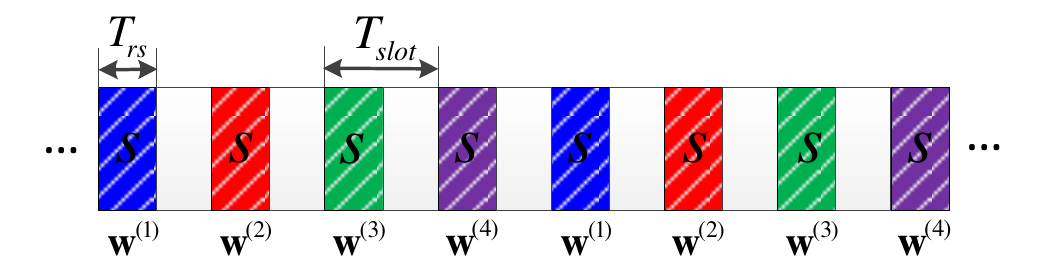}
  \\
  (b)
\end{subfigure}
\caption{An illustration of the 3-step procedure for designing RS transmission strategy using beamformers: (a) Partitioning of sub-intervals with $M=4$ and the desired patterns in~\eqref{Optimal_pattern}; (b) Frame allocations using the $4$ beamformers for RS transmission.}
\label{fig:construct_beam_example}
\vspace{-1em}
\end{figure}

The proposed procedure offers flexibility in synthesising the beamformers to provide the best approximation to $G^*(\phi)$, by choosing different angular partitionings (the $\Omega_m$s and the value of $M$). Such a flexibility is beneficial when $\alpha(\phi)$ is not uniform in $\phi$,
since by slicing the whole angular interval into $M>1$ sub-intervals, a smaller variation of the pathloss $\alpha(\phi)$ can be achieved within each sub-interval and a better approximation to the desired pattern can be achieved using practical beam synthesis techniques.

Even if $\alpha(\phi)$ is uniform such that an omnidirectional transmission within $\Omega$ is desirable, the offered flexibility is beneficial. For instance, consider the two sets of beamformers illustrated in  Fig.~\ref{fig:system_model} (b) and (c), with $M=4$ and $M=1$, respectively. In principle, they can achieve comparable approximations to \eqref{Optimal_pattern_condition} and thus comparable BS discovery performance if appropriate hardware resources are available. However, with limited hardware resources and imperfect approximations to \eqref{Optimal_pattern_condition}, their real performance can be significantly different. As we will show later in Section~\ref{sec:simulation}.~C, the proposed strategy allows us to design different codebooks to attain the best BS discovery performance under different hardware limits such as per-antenna power constraints.

\section{Numerical Results}\label{sec:simulation}
In this section, we demonstrate the performance of BS discovery under the beamforming strategy proposed in Section~\ref{sec:main_theorem}. Throughout this section, ULAs with $N_T = 32$  and $N_R=16$ omnidirectional antenna elements  at the BS and at the UE are considered. The intended coverage angular space at the BS is a $60^\circ$ sector with $\phi \in [-30^{\circ}, 30^{\circ}]$. Antenna spacing at both the BS and UEs is fixed at half a wavelength of the carrier frequency.

We consider two different coverage topologies. For the first topology, the entire  $60^\circ$ sector is open without blockage in all directions, therefore $R_{\text{th}}(\phi)$ is set to a constant for all directions: $R_{\text{th}}(\phi) = R_{\text{th}}$. Parameter $\alpha$ is then calculated as described in Section IV, e.g., $\alpha = \frac{P_T}{\sigma^2}\frac{G_T^{\max}G_R^{\max}}{\text{SNR}_{th}}\frac{W_{rs}}{W}$. For the second topology, we consider that half to the $60^\circ$ sector, e.g., $\phi \in [-30^\circ,0^\circ]$, is blocked by a nearby building, as illustrated in Fig.~\ref{fig:pathloss}. The coverage range of these blocked directions is thus shorter than that of the open directions. For simplicity, $\alpha(\phi)$ for the blocked directions is set to $\alpha/2$.
Multipath fading channels are considered in our simulation and are generated using the same parameters for Fig.~\ref{eq:upper_bound} (See Section III). In particular, we consider that the AoDs from the BS are drawn uniformly in $[-30^{\circ}, 30^{\circ}]$, and that $\sum_{q=1}^{Q}\mathbb{E}\{|g_{l,q}|^2\}=1/\alpha(\phi)$. Details of the relevant system parameters are provided in Table~\ref{Tab:paramter}.

\begin{figure}[t]
\centering
\includegraphics[width=0.4\textwidth]{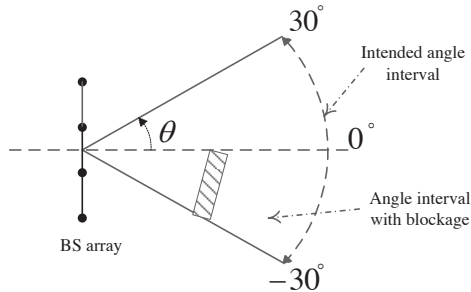}
\caption{A $60^\circ$ sector: half of the angular coverage interval has blockage.}
\label{fig:pathloss}
\vspace{-1 em}
\end{figure}

\begin{table}[t]\centering
\caption{Simulation parameters}\label{Tab:paramter}
\small{
\begin{tabular}{|l|l|}
\hline
Duration of  RS:   & $T_{rs} =10 \mu s$ \\
\hline
 Duration of slot:  &  $T_{slot} = 0.5 ms$\\
 \hline
 RS Bandwidth:  & $W_{rs} = 10$ MHz\\
\hline
Sampling interval: & $T_s = 1/W_{rs}$ \\
\hline
 Data bandwidth: &  $W=1$ GHz\\
 \hline
 Minimum data-rate:  & $R_{th}=10$ Mbps \\
\hline
 False alarm target: &  $P_{FA} = 0.001$\\
\hline
Fraction of time for downlink data trans. &  $\rho=0.4$\\
\hline
Maximum beamforming Gain & $G_T^{\max}=N_T$, $G_R^{\max}=N_R$\\
\hline
\end{tabular}
}
\vspace{-1em}
\end{table}

\subsection{Beam pattern synthesis methods}\label{sec:beamforming}
Various techniques may be used to synthesise beamformers to approximate the ones defined in \eqref{Optimal_pattern}.
In one category of such techniques, only the phase of the signal at each antenna port is controlled~\cite{abu2014synchronization,6600706,Xiao2016,raghavan2016beamforming}. The transmitted signals across all antennas thus have a constant modulus. The other category controls both the phase and magnitude of the signal emitted from each antenna~\cite{alkhateeb2014channel, lebret1997antenna} and as such, beamformers constructed are expected to provide better approximations.  However, extra hardware cost is required to implement such beamformers. In this subsection, we present beamformer synthesis methods for both the categories, namely, \emph{Constant modulus}~(CM) and \emph{Variable modulus}~(VM), which will be used in the subsequent parts of our numerical experiments. The main purpose, though, is not to promote any specific beam synthesis method; rather, it is to demonstrate the effectiveness of the proposed beam synthesis procedure under different hardware cost constraints.

We adopt the optimisation-based approach in~\cite{abu2014synchronization} and \cite{raghavan2016beamforming} to synthesise CM beamformers to approximate the desired beam patterns. The method in~\cite{abu2014synchronization} minimises a weighted mismatch between the desired pattern and the synthesised pattern by optimising phase coefficients of the antennas, while the method in~\cite{raghavan2016beamforming} aims to maximise the minimum beamforming gain within a desirable interval. Both the two methods are adopted because of their flexibility in controlling the beamwidth and their excellent beam broadening ability.

To flexibly control the beamwidth, in synthesising the VM beamformers, we adopt the same optimisation objective as in~\cite{abu2014synchronization}, and optimise both the phases and the magnitudes of the antennas. The beamformers to be optimised takes the form of $\mathbf{w}^{(m)} = \left[w_m(1)e^{i\varphi(1)},\ldots, w_m(N_T)e^{i\varphi(N_T)}\right]^T$, where $w_m(n)$ are real-valued variables such that $||\mathbf{w}^{(m)}||_2^2=1$. An extra constraint, i.e., $w_m(n) = w_m(N_T-n-1)$, is imposed such that the synthesised pattern is symmetric with respect to the centre of the beam and to reduce the dimension of the optimisation. A genetic algorithm is applied to solve the corresponding global optimsation problems~\cite{goldberg2000genetic}. (The method in~\cite{raghavan2016beamforming}  cannot be applied since the formulas provided are only for the CM case.)

The Subarray Method (SM) recently presented in~\cite{Xiao2016} is also considered since close-form beamformer formulas are provided (thus no optimisation is required).

Fig.~\ref{fig:beam_pattern_CM}-\ref{fig:beam_pattern_SM} present the patterns of the synthesised beamformers using the methods mentioned above. Three choices of $M$ are considered, with $M=\{1,2,4\}$, where the desired patterns are indicated by dash lines. Note that in Fig.~\ref{fig:beam_pattern_CM}, we have termed the method in~\cite{abu2014synchronization} as CM-1 and the method in~\cite{raghavan2016beamforming} as CM-2. It can be seen that both the CM and VM methods can provide good approximations to the desired pattern. The SM beamformers have noticeable smaller beamforming gain at beam edges, which may become the bottleneck according to Proposition~\ref{proposition_optimal_beam}.

For all the 9 sets of synthesised beamformers, VM=1 is found to have the highest minimum average beamforming gain within the sector. In the following, these synthesised beamformers are adopted by the proposed RS transmission strategy and their relative performances are compared.

\begin{figure}[t]
\centering
\begin{subfigure}%{0.5\textwidth}
  \centering
  \includegraphics[width=0.3\linewidth]{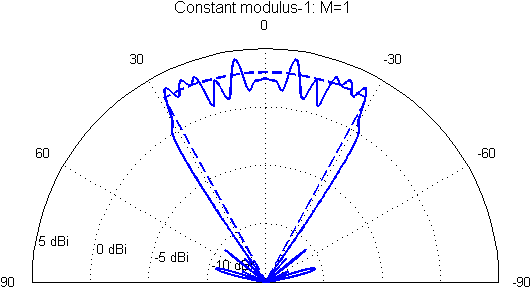}
\end{subfigure}
\begin{subfigure}%{0.5\textwidth}
  \centering
  \includegraphics[width=0.3\linewidth]{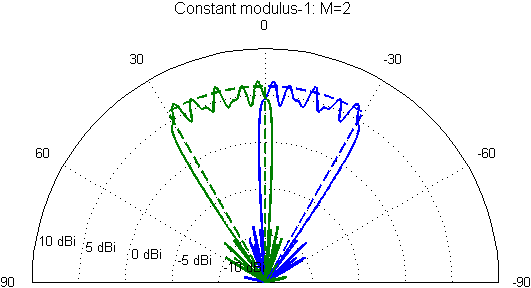}
\end{subfigure}
\begin{subfigure}%{0.5\textwidth}
  \centering
  \includegraphics[width=0.3\linewidth]{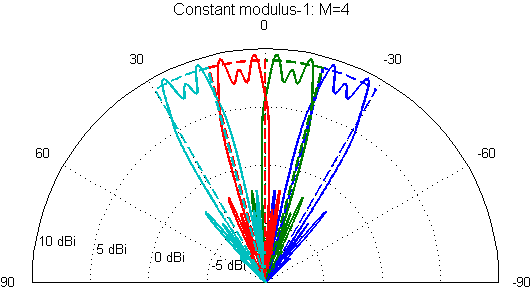}
\end{subfigure}
\begin{subfigure}%{0.5\textwidth}
  \centering
  \includegraphics[width=0.3\linewidth]{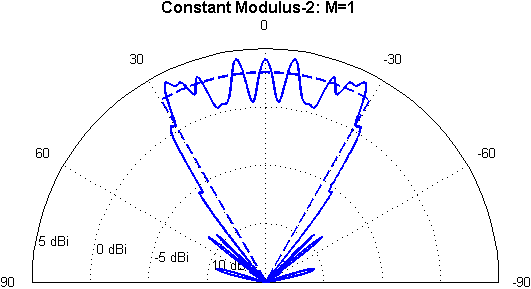}
\end{subfigure}
\begin{subfigure}%{0.5\textwidth}
  \centering
  \includegraphics[width=0.3\linewidth]{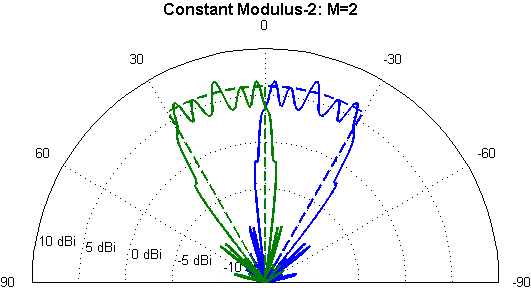}
\end{subfigure}
\begin{subfigure}%{0.5\textwidth}
  \centering
  \includegraphics[width=0.3\linewidth]{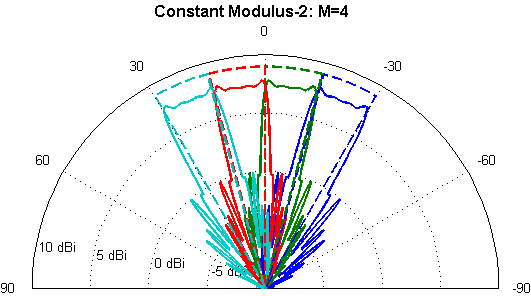}
\end{subfigure}
\caption{CM beam patterns with $M\in \{1,2,4\}$: $N_T = 32$. Solid lines: synthesised patterns; Dashed lines: desirable patterns. }
\label{fig:beam_pattern_CM}
\vspace{-1em}
\end{figure}

\begin{figure}[t]
\centering
\begin{subfigure}%{0.5\textwidth}
  \centering
  \includegraphics[width=0.3\linewidth]{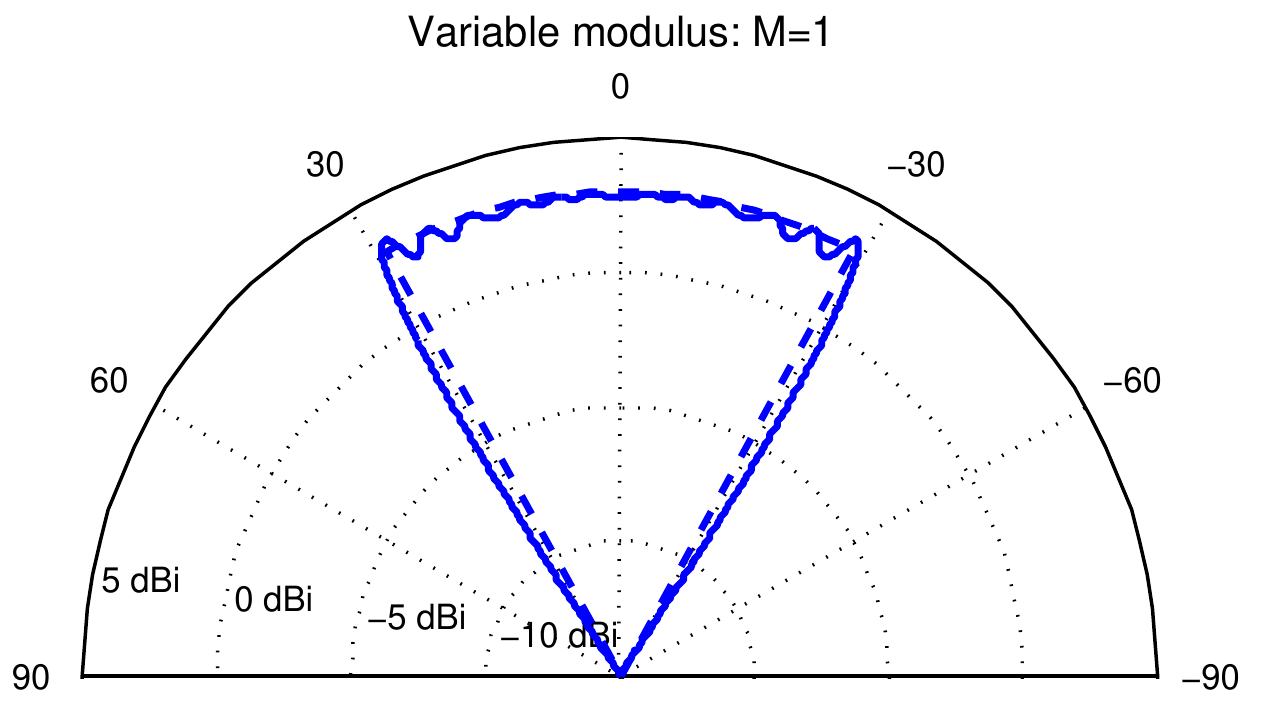}
\end{subfigure}
\begin{subfigure}%{0.5\textwidth}
  \centering
  \includegraphics[width=0.3\linewidth]{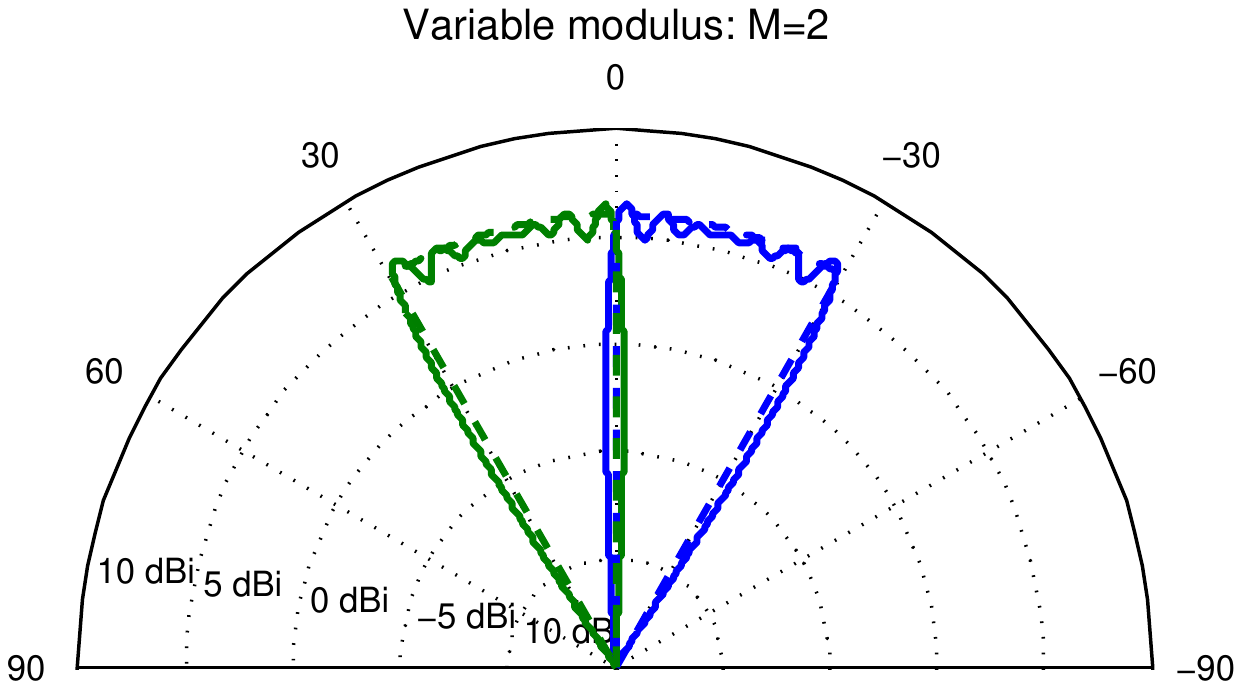}
\end{subfigure}
\begin{subfigure}%{0.5\textwidth}
  \centering
  \includegraphics[width=0.3\linewidth]{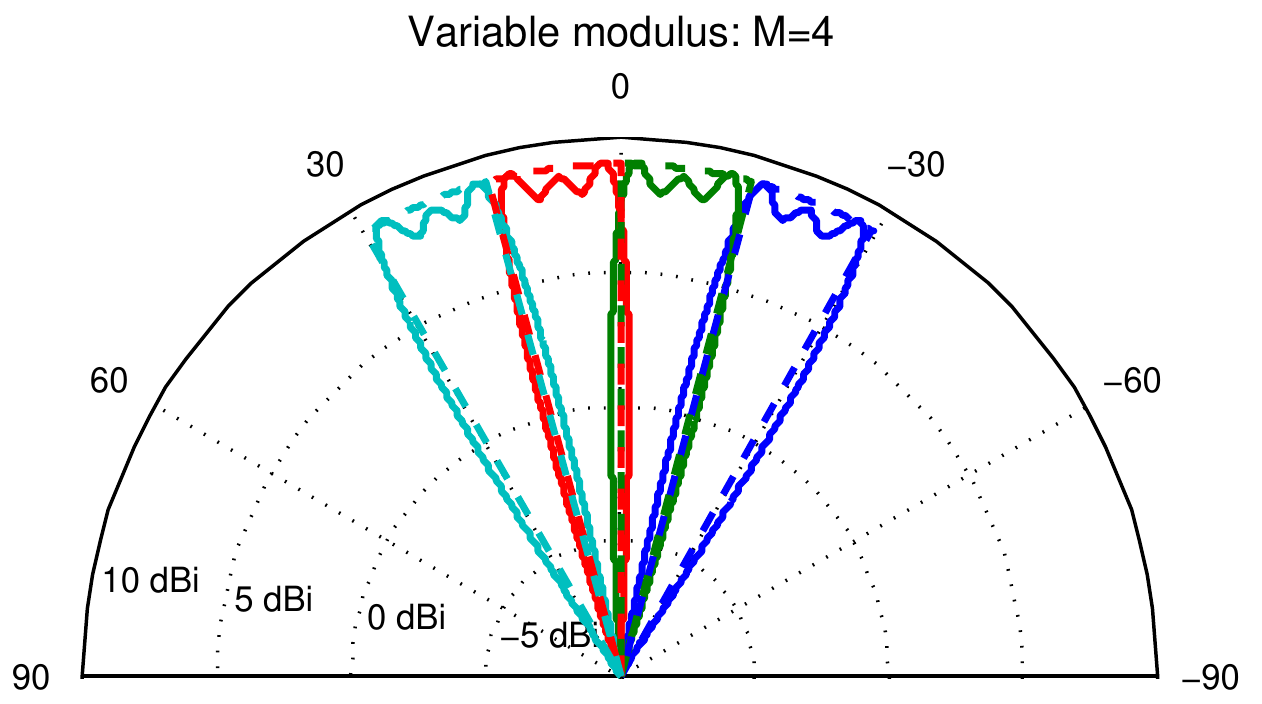}
\end{subfigure}
\caption{VM beam patterns with $M\in \{1,2,4\}$: $N_T = 32$. Solid lines: synthesised patterns; Dashed lines: desirable patterns. }
\label{fig:beam_pattern_VM}
\vspace{-1em}
\end{figure}

\begin{figure*}[t]
\centering
\begin{subfigure}%{0.5\textwidth}
  \centering
  \includegraphics[width=0.3\linewidth]{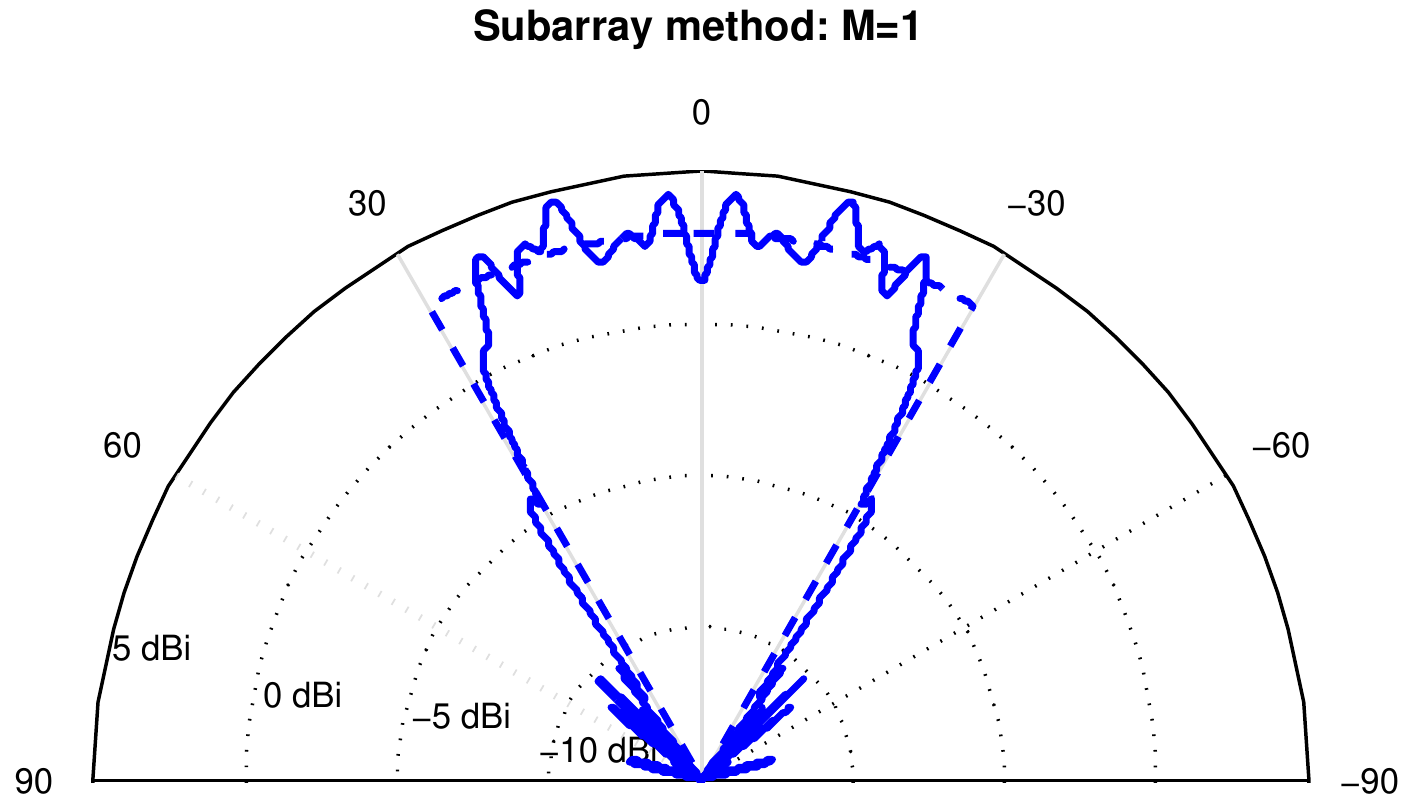}
\end{subfigure}
\begin{subfigure}%{0.5\textwidth}
  \centering
  \includegraphics[width=0.3\linewidth]{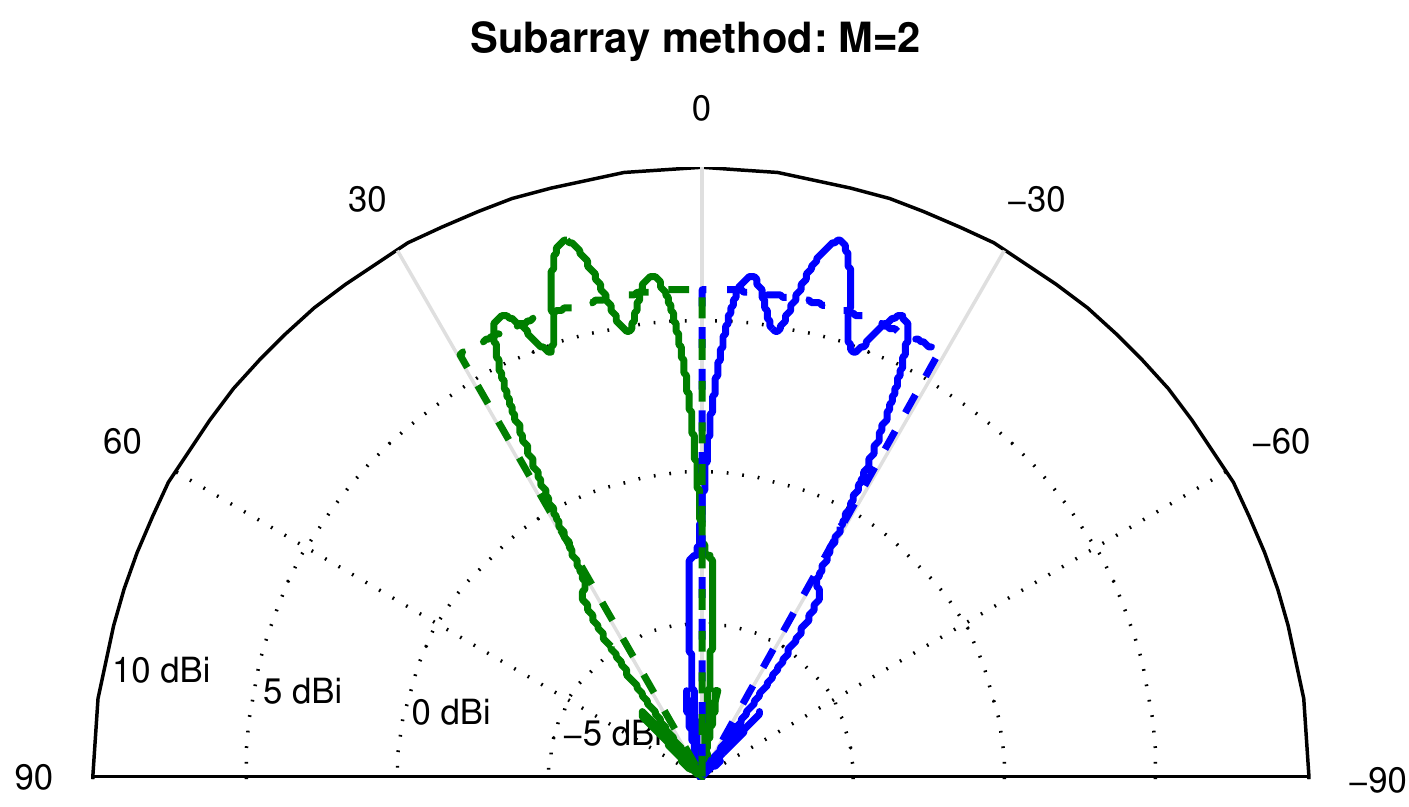}
\end{subfigure}
\begin{subfigure}%{0.5\textwidth}
  \centering
  \includegraphics[width=0.3\linewidth]{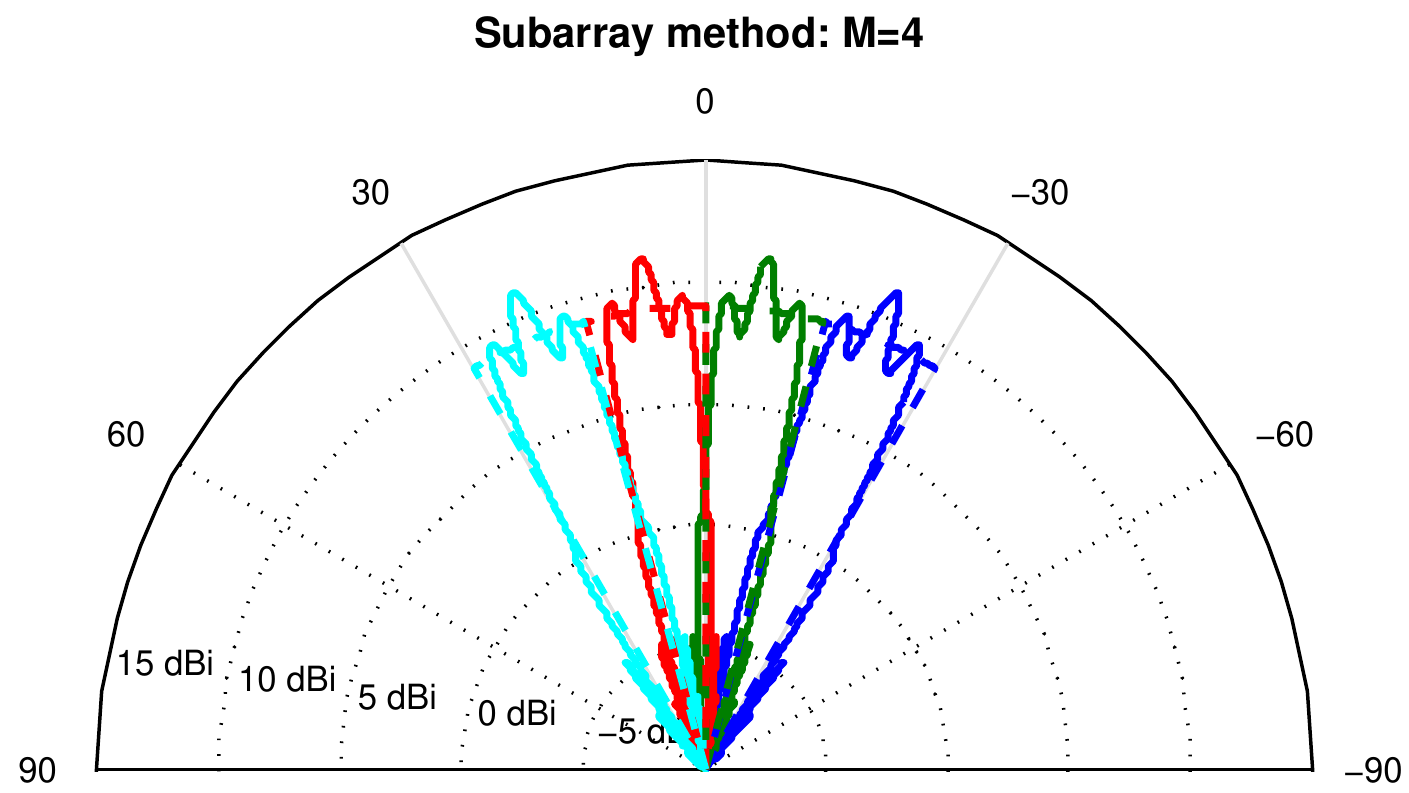}
\end{subfigure}
\caption{SM~\cite{Xiao2016} beam patterns with $M\in \{1,2,4\}$: $N_T = 32$. Solid lines: synthesised patterns; Dashed lines: desirable patterns.}
\label{fig:beam_pattern_SM}
\vspace{-1 em}
\end{figure*}

\begin{figure}[t]
\centering
\includegraphics[width=0.6\textwidth]{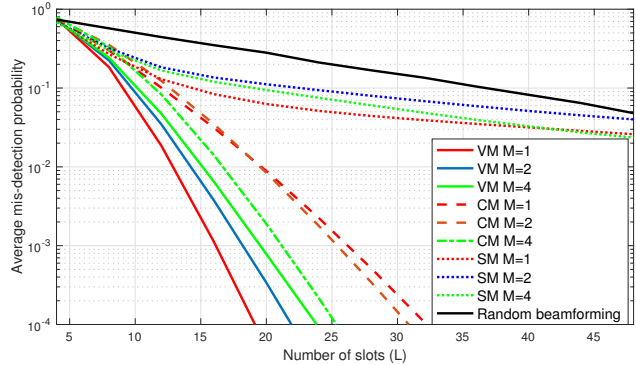}
\caption{Average miss-detection probability for an open $60^\circ$ sector: no blockage.  }
\label{fig:Analog_vs_digital}
\vspace{-1 em}
\end{figure}

\subsection{BS discovery performance without blockage}\label{sec:no_blockage}
In this subsection, we consider the first topology, where no blockage is present in the $60^{\circ}$ sector. Since in this case, $\alpha(\phi)$ is a constant for all directions in the sector, it can be seen from \eqref{Eq:number_beamformer} that $J_m = \frac{J}{M}$. 

Fig.~\ref{fig:Analog_vs_digital} shows the average miss-detection probability versus the number of slots ($L$) that UEs use for BS discovery, for the VM, CM and SM beamformers. When $M=\{1, 2\}$, CM-1 method is adopted as CM-1 outperforms CM-2. At $M=4$, CM-2 is used for a similar reason. RS transmission using random beamformers as considered in~\cite{barati2014dreictional} is also included as a baseline. The average miss-detection probability is obtained by averaging all possible directions within the sector and over $500$ channel realisations for each $L$.

As expected, RS transmission with VM=1 has the best BS discovery performance since it yields the highest minimum beamforming gain over all directions. The VM beamformers generally outperform the CM beamformers since the VM beamformers are more capable of approximating the desired patterns provided in \eqref{Optimal_pattern}. RS transmissions with the SM beamformers have much worse performance than that with the CM and VM beamformers, when $L$ is larger than $15$. This is because for the SM beamformers, significantly lower beamforming gain is observed at beam edges, as illustrated in Fig.~\ref{fig:beam_pattern_SM}.

In addition, the BS discovery performance exhibits different behaviours for VM beamformers and CM beamformers when different codebook sizes are considered. For the VM method, as shown in Fig.~\ref{fig:Analog_vs_digital}, $M=1$ achieves the best detection performance for all the values of $L$ examined. On the contrary, for the CM method, $M=4$ provides the best performance.
These results demonstrate that having a flexible choice of $M$ in the design, as we propose, is necessary and beneficial as each beam synthesis method may have its own 'best' $M$.

\subsection{BS discovery performance with blockage}\label{sec:blockage}
\begin{figure*}[t]
\center
\begin{minipage}{0.35\textwidth}
    \begin{center}
        \includegraphics[width=1\linewidth]{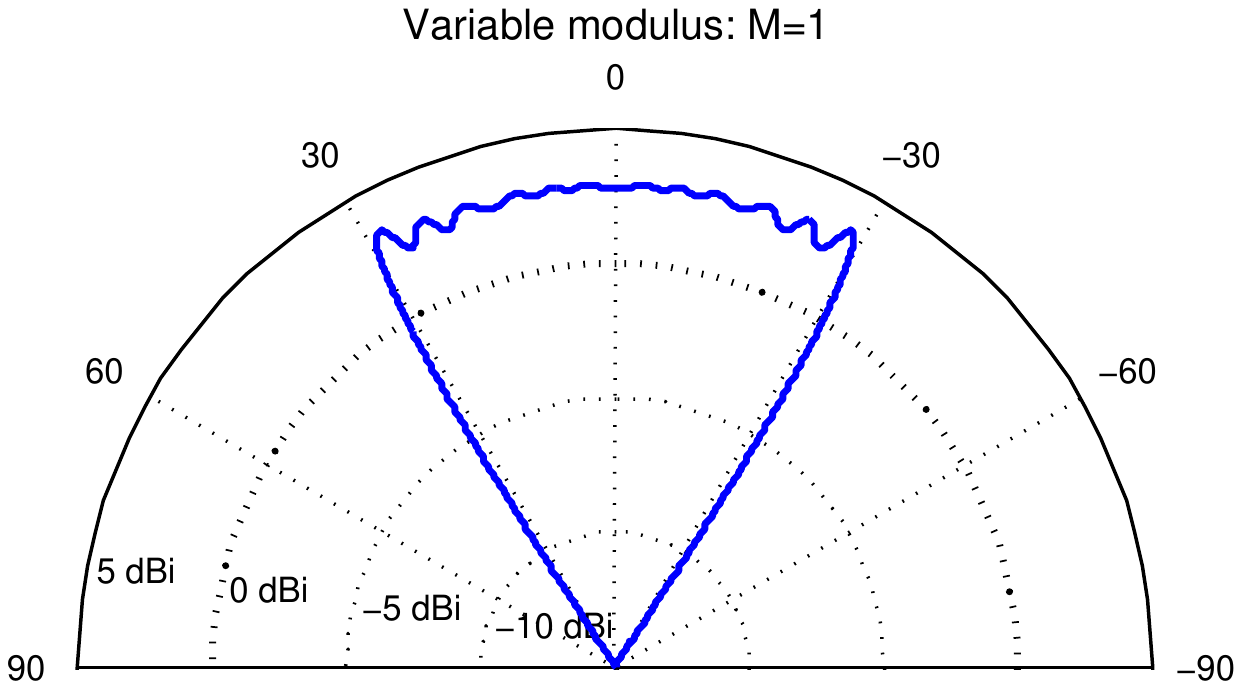}
	    \\
\vspace{+0.3cm}
        \includegraphics[width=1\linewidth]{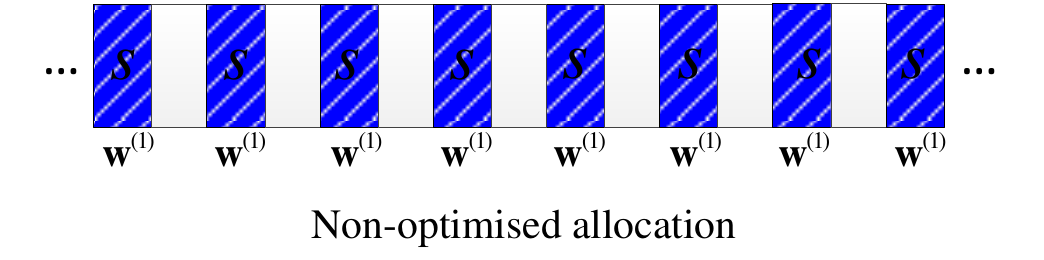}
        \\
        (a)
    \end{center}
\end{minipage}
\vspace{+0.3cm}
\begin{minipage}{0.35\textwidth}
    \begin{center}
        \includegraphics[width=1\linewidth]{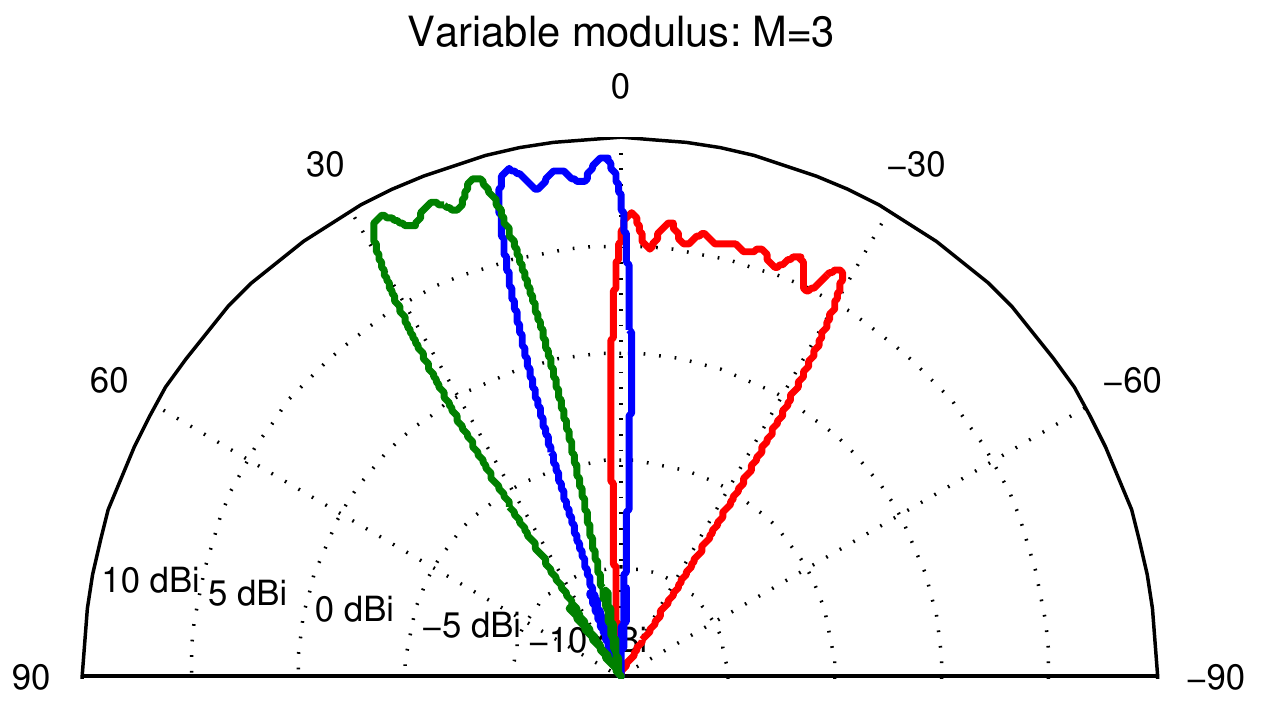}
        \\
        \vspace{+0.2cm}
        \includegraphics[width=1\linewidth]{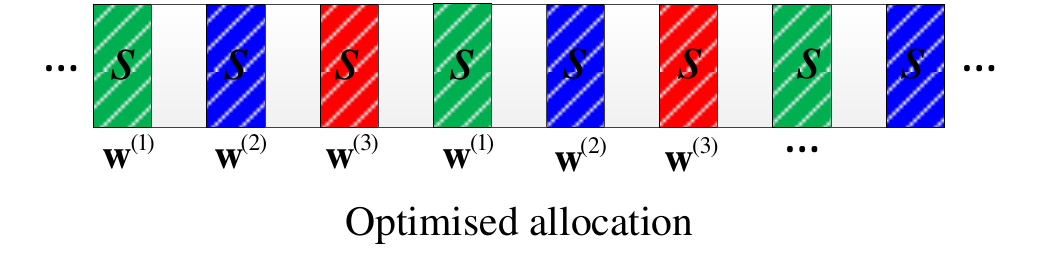}
        \\
        (c)
    \end{center}
\end{minipage}
\vspace{+0.3cm}
\begin{minipage}{0.35\textwidth}
    \begin{center}
        \includegraphics[width=1\linewidth]{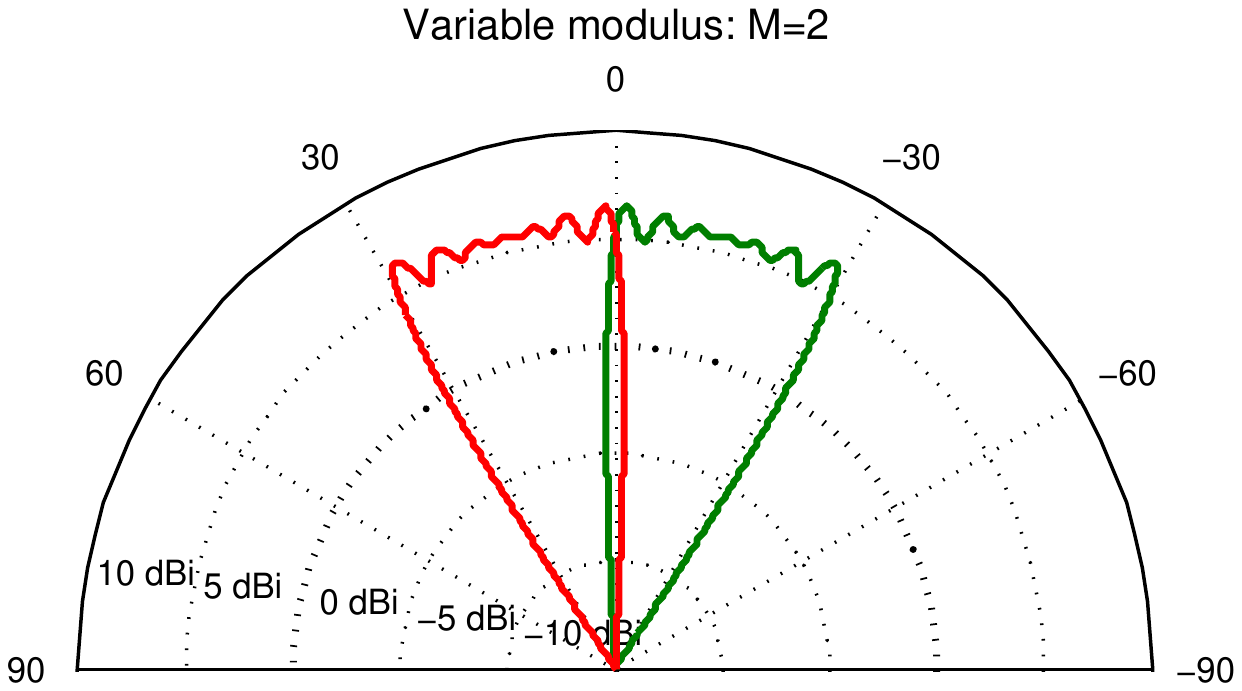}
        \\
        \vspace{+0.5cm}
        \includegraphics[width=1\linewidth]{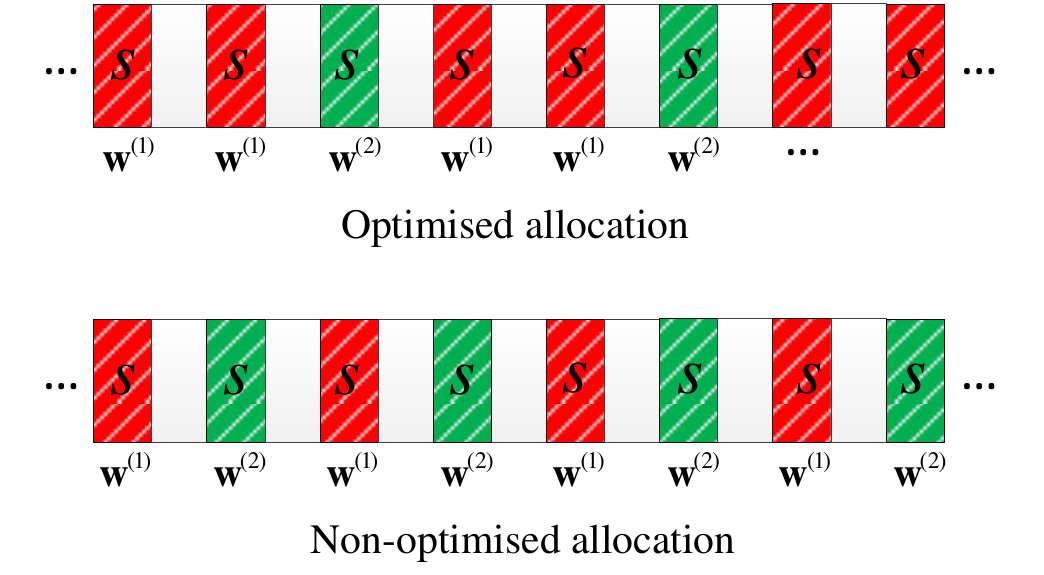}
        \\
        (b)
    \end{center}
\end{minipage}
\begin{minipage}{0.35\textwidth}
    \begin{center}
        \includegraphics[width=1\linewidth]{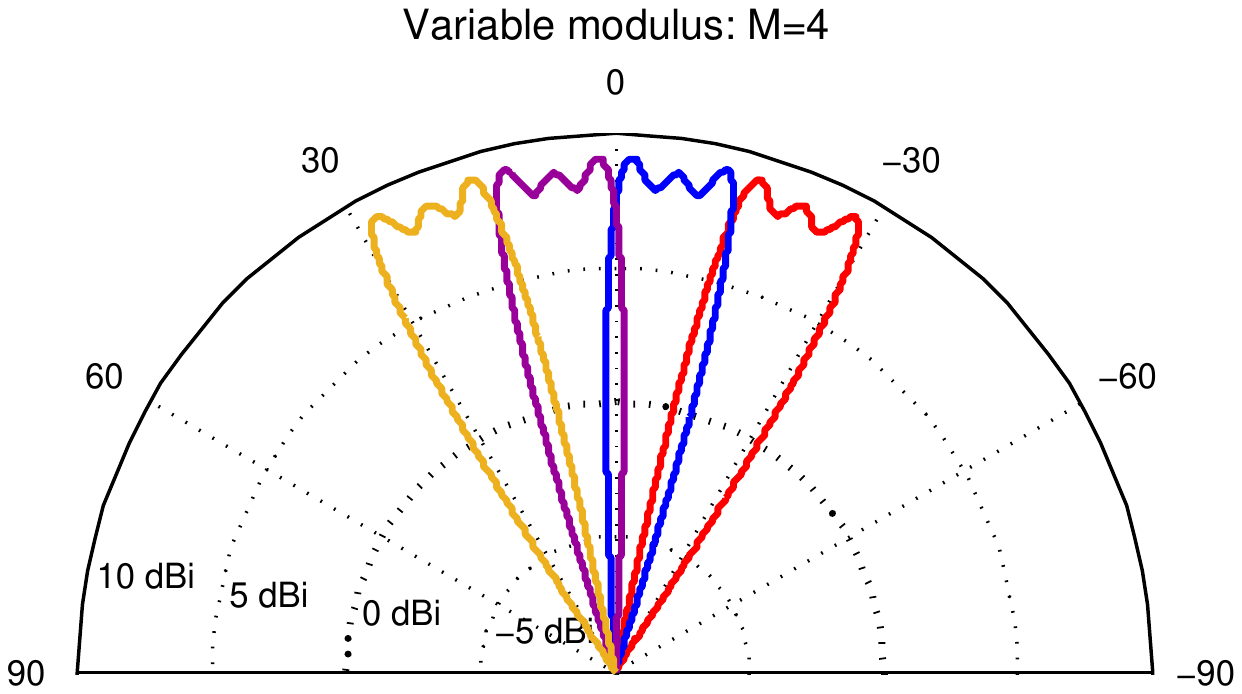}
        \\
        \vspace{+0.3cm}
        \includegraphics[width=1\linewidth]{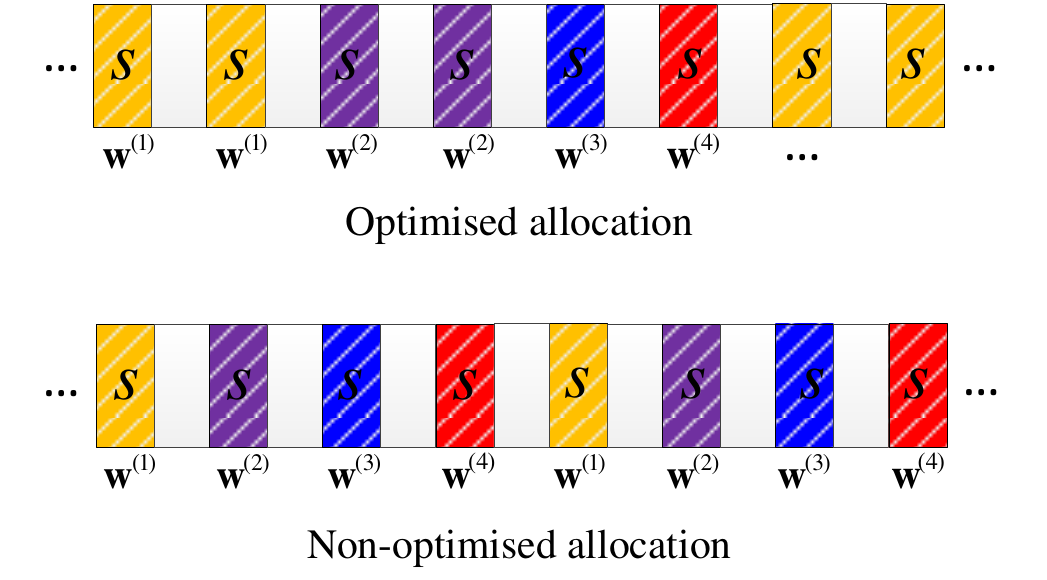}
        \\
        (d)
    \end{center}
\end{minipage}
\caption{Beam patterns and slot allocation for the synthesised beamformers using VM method. \emph{Optimised} slot allocation is obtained by applying \eqref{Eq:number_beamformer}, while \emph{Non-optimised} slot allocation takes an equal slot allocation to the $M$ beamformers.}
\label{fig:beam_pattern_asym}
\vspace{-1.0 em}
\end{figure*}

In this subsection, we consider the second topology with blockage, as illustrated in Fig.~\ref{fig:pathloss}.
For this topology, according to \eqref{Optimal_pattern_condition}, a higher average beamforming gain is required at the open directions ($[0^{\circ}, 30^{\circ}]$) compared to the directions with blockage ($[-30^{\circ}, 0^{\circ}]$). To approximate \eqref{Optimal_pattern_condition}, three values of $M$, i.e., $M=2,3,4$ are examined. It is noted that we do not have $M=1$ for this example due to the difficulty in synthesising a single beamformer to match the desired non-uniform beam pattern. The synthesised beam patterns, along with optimised slot allocations, are illustrated in Fig.~\ref{fig:beam_pattern_asym}. We refer to these as the \emph{optimised} schemes. (For brevity, we only present RS transmission designs using VM beamformers and their corresponding detection performance.)

The performance of BS discovery from the optimised beamforming strategies are compared to a set of baseline strategies. For these baselines, the beamformers are obtained in the same way as those for the optimised cases. However, an equal slot allocation is used. It is easy to check that the resultant average beam pattern does not satisfy condition \eqref{Optimal_pattern_condition}. In Fig.~\ref{fig:beam_pattern_asym}, these baselines are also illustrated and we refer them as \emph{non-optimised} schemes. RS transmission using random beamformers as considered in~\cite{barati2014dreictional} is also included as a baseline.

Fig.~\ref{fig:blocked_model} presents the corresponding BS discovery performance as a function of the number of slots used for BS discovery ($L$). 
Clearly, the optimised RS transmissions using the beamformers constructed according to \eqref{Optimal_pattern_condition} significantly outperform all the baselines. For a reasonable UE searching time of $5$ ms (e.g., $L=10$), the optimised RS transmissions provide orders of magnitude improvement over the RS transmission with random beamforming.
For a targeted average miss-detection rate, e.g., at $10^{-3}$, the optimised RS transmission with $M=2$ requires a searching time $L=12$, which is $20\%$ faster than that required by the best non-optimised baseline ($M=1$).
These results clearly demonstrate the benefit of designing the RS transmissions that take into account asymmetric coverage of a mm-wave BS.

\begin{figure}[t]
\centering
\includegraphics[width=0.6\textwidth]{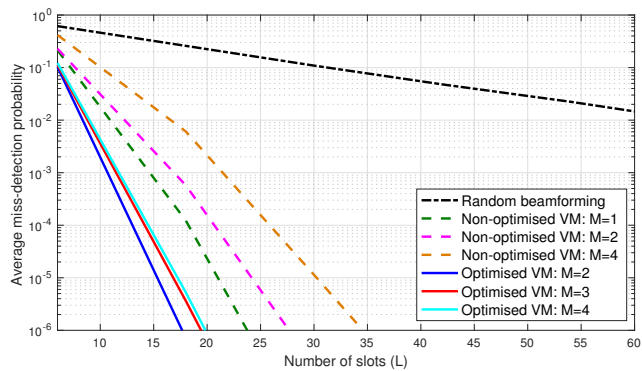}
\caption{Average miss-detection probabilities for a half-blocked $60^\circ$ sector (as illustrated in Fig.\ref{fig:pathloss}).}
\label{fig:blocked_model}
\vspace{-0.6 cm}
\end{figure}

\subsection{The impact of per-antenna power constraint}
Up until now, we have assumed that the dynamic ranges of the power amplifiers (PAs) for the BS antennas are large enough to accommodate the VM beamformers. However, due to the poor power efficiency of the PAs in mm-wave bands, this may not be true since each PA may be designed to operate at full power in order to provide sufficient transmit power~\cite{huang2011millimeter}. In this context, per-antenna power constraints are important, and may influence the implementation of the VM beamformers that require variable amplitudes at different antennas.

In the following, we demonstrate the impact of per-antenna power constraints on the construction of RS transmit beamformers and on the performance of BS discovery. The first topology without blockage is considered.
In the simulation,
we introduce a normalised per-antenna power constraint, denoted by $\beta$, as the ratio of per-antenna power limit to the maximum total transmit power $P_T$.
We consider that $\beta$ is in the range $\beta \in [1/N_T,1]$. In this case, for the CM beamformers, the transmit power at each antenna is $P_T/N_T$. For the VM beamformers, the beamforming coefficients are scaled up/down by the same factor such that both the per-antenna and the total power constraint are satisfied. In particular, when $\beta=1/N_T$, because of the uniform scaling, only the antenna with the largest beamforming weight transmits with power $P_T\beta = P_T/N_T$, all other antennas are transmitting with smaller powers. This means that the total transmit power is less than $P_T$, i.e., there is a loss in the total transmission power. As $\beta$ increases, the power loss becomes smaller and ultimately reduces to zero for sufficiently large $\beta$s.

\begin{figure}[t]
\centering
\includegraphics[width=0.6\textwidth]{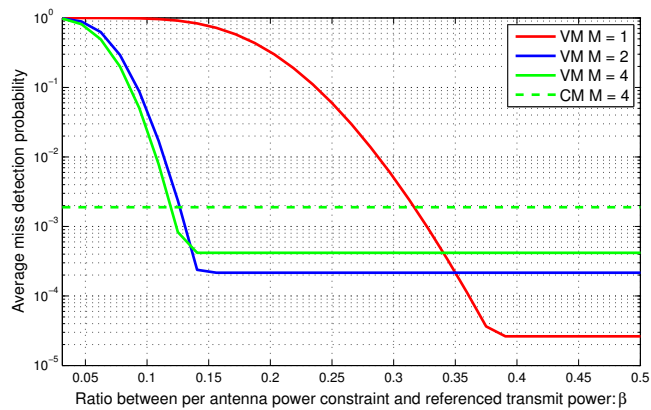}
\caption{Average miss-detection probability with different beamfomers and per antenna power constraint: $L=20$.  }
\label{fig:per_antenna}
\vspace{-1 em}
\end{figure}

Fig.~\ref{fig:per_antenna} shows the average miss-detection probability versus per-antenna power constraint $\beta$.
It can be seen that when $\beta$ is small (e.g., $\beta<0.1$), the VM beamformers perform significantly worse than the CM beamformers due to the power loss. Increasing the dynamic range of the power amplifiers will favor the VM implementation.  However, this comes at the price of increased cost and difficulty in implementing mm-wave transceivers. The trade-off between achievable performance and the hardware cost is clearly shown via these numerical results.

Moreover, the results also demonstrate the benefit of the proposed beamformer construction method due to a flexible choice of $M$. For instance, under a strict per-antenna power constraint, e.g., $\beta=0.15$, the VM beamformer with $M=2$ outperforms the alternative VM beamformers with $M=1$ and $M=4$. In the absence of per-antenna power constraint, the VM beamformer with $M=1$ performs the best.
However, this benefit comes at the price of $5$ dB dynamic range increase as compared with that required by $M=2$.

\section{Conclusions}\label{sec:con}
In this paper, we have developed an analytical framework for mm-wave BS discovery and proposed an effective RS transmission strategy based on sequential spatial scanning using beamforming. We have established the relationship between the performance of BS discovery and system parameters, including the beamformers used for RS transmission. Using the method of large deviations, we have identified the desirable beam patterns for RS transmission, and shown that the desirable beam patterns are asymptotically optimal and yield the minimum average miss-discovery probability for UEs in a targeted detectable region. To approximate the desirable beam patterns, we have proposed a systematic approach to construct the beamforming codebook, which is flexible in choosing the beam widths and thus the size of the corresponding codebook. Numerical results have demonstrated the effectiveness of the proposed method.

\appendices
\section{Proof of Proposition~\ref{Proposition:GLRT}}\label{App_GLRT_detector}
In this proof, for notational simplicity, we drop the dependency of all metrics on $\tau$ when no ambiguity is caused.

Since all entries of $\mathbf{Z}$ are i.i.d. complex Gaussian with zero mean and variance $\sigma^2$, the PDF of $\mathbf{Z}$ can be represented as $p(\mathbf{Z}) = \frac{1}{(\pi \sigma^2)^N} \exp{\left\{-\frac{\|\mathbf{Z}\|_F^2}{\sigma^2}\right\}}$, where $N\doteq N_RN_SL$.
Therefore, it can be obtained that
$p\left( \mathbf{Y}|{\cal H}_0; \sigma_0^2 \right) = \frac{1}{(\pi \sigma_0^2)^N} \exp{\left\{-\frac{\|\mathbf{Y}\|_F^2}{\sigma_0^2}\right\}}$
and
$p\left(\mathbf{Y}|{\cal H}_1;\mathbf{h},\sigma_1^2 \right) = \frac{1}{(\pi \sigma_1^2)^N} \exp{\left\{-\frac{\|\mathbf{Y}-\mathbf{h}\mathbf{s}^T\|_F^2}{\sigma_1^2}\right\}}$.
It can then be shown that under ${\cal H}_0$, the ML estimate of $\sigma_0^2$ is
{\small\begin{equation}\label{ML_sigma0}
\hat{\sigma}_0^2 = \frac{||\mathbf{Y}||_F^2}{N} =  \frac{1}{N}\sum_{l=1}^L \|\mathbf{Y}_l\|_F^2,
\end{equation}}
and under ${\cal H}_1$, the ML estimate of $\mathbf{h}$ and $\sigma_1^2$ are:
{\small\begin{equation}\label{ML_h1}
\hat{\mathbf{h}} =  \frac{1}{\|\mathbf{s}\|_2^2}\mathbf{Y}\mathbf{s}^*,
\end{equation}}
and
{\small\begin{equation}\label{ML_sigma1_a}
\hat{\sigma}_1^2 = \frac{\|\mathbf{Y}-\hat{\mathbf{h}}\mathbf{s}^T\|_F^2}{N},
\end{equation}}
respectively, where $\mathbf{s}^* = (\mathbf{s}^\dagger)^T$.

Substituting \eqref{ML_h1} into \eqref{ML_sigma1_a}, it can be obtained that
{\small
\begin{align}\label{ML_sigma1}
\hat{\sigma}_1^2 & = \frac{1}{N}\left(\left\|\mathbf{Y} - \frac{1}{\|\mathbf{s}\|_2^2}\mathbf{Y}\mathbf{s}^*\mathbf{s}^T\right\|_F^2\right) =\frac{1}{N} \mathrm{Tr}\left\{\left(\mathbf{Y} - \frac{1}{\|\mathbf{s}\|_2^2}\mathbf{Y}\mathbf{s}^*\mathbf{s}^T\right)\left(\mathbf{Y} - \frac{1}{\|\mathbf{s}\|_2^2}\mathbf{Y}\mathbf{s}^*\mathbf{s}^T\right)^\dagger \right\}\nonumber\\
& = \frac{1}{N}\left(\|\mathbf{Y}\|_F^2 - \frac{1}{\|\mathbf{s}\|_2^2}\|\mathbf{Y}\mathbf{s}^*\|_F^2\right) = \frac{1}{N}\sum_{l=1}^L \left( \|\mathbf{Y}_l\|_F^2 - \frac{1}{\|\mathbf{s}\|_2^2}\|\mathbf{Y}_l\mathbf{s}^*\|_2^2 \right).
\end{align}
}

Substituting \eqref{ML_sigma0}, \eqref{ML_h1} and \eqref{ML_sigma1} into \eqref{GLRT_statistic}, the test statistics $L'_G(\tau)$ can be represented as:
{\small\begin{equation}\label{GLRT_sta}
 L'_G(\tau) = \left( \frac{\hat{\sigma}_0^2}{\hat{\sigma}_1^2}\right)^N = \left(\frac{ \sum_{l=1}^L\|\mathbf{Y}_{l,\tau}\|_F^2}{\sum_{l=1}^L \left(\|\mathbf{Y}_{l,\tau}\|_F^2 - \frac{1}{\|\mathbf{s}\|_2^2}\|\mathbf{Y}_{l,\tau}\mathbf{s}^*\|_2^2\right)} \right)^N.
\end{equation}}
Equivalently, the test \eqref{GLRT_statistic} can be re-written as \eqref{GLRT_sta_log}, 
%\begin{equation}
%L_G(\tau) = \frac{\sum_{l=1}^L \frac{1}{\|\mathbf{s}\|_2^2}\|\mathbf{Y}_{l,\tau}\mathbf{s}^*\|_2^2}{\sum_{l=1}^L \left(\|\mathbf{Y}_{l,\tau}\|_F^2 - \frac{1}{\|\mathbf{s}\|_2^2}\|\mathbf{Y}_{l,\tau}\mathbf{s}^*\|_2^2\right)} \mathop{\gtrless}_{{\cal H}_0}^{{\cal H}_1}\gamma,
%\end{equation}
where $\gamma = \left(\gamma'\right)^{1/N}-1>0$ is the test threshold.

\section{Proof of Proposition~\ref{Proposition_GLRT_sta}}\label{App_pGLRT}
We first show that random variables $U_l\triangleq  \frac{1}{\|\mathbf{s}\|_2^2} ||\mathbf{Y}_{l}\mathbf{s}^*||_2^2$ and $V_l \triangleq ||\mathbf{Y}_{l}||_F^2 -  \frac{1}{\|\mathbf{s}\|_2^2} ||\mathbf{Y}_{l}\mathbf{s}^*||_2^2$
%\begin{equation}\label{Eq:U}
%U_l\triangleq  \frac{1}{\|\mathbf{s}\|_2^2} ||\mathbf{Y}_{l}\mathbf{s}^*||_2^2
%\end{equation}
%and
%\begin{equation}\label{Eq:V}
%V_l \triangleq ||\mathbf{Y}_{l}||_F^2 -  \frac{1}{\|\mathbf{s}\|_2^2} ||\mathbf{Y}_{l}\mathbf{s}^*||_2^2
%\end{equation}
are statistically independent, where we recall that $\mathbf{s}^* = (\mathbf{s}^\dagger)^T$.

Under event ${\cal H}_1$, it can be obtained that
{\small\begin{equation}
\mathbf{Y}_{l} = \mathbf{h}_l\mathbf{s}^T + \mathbf{Z}_l.
\end{equation}}
Denote $P_T$ as the average transmit power with $\|\mathbf{s}\|_2^2 = N_sP_T$ and $\bar{\mathbf{z}}_{l,N_s} = \frac{1}{\sqrt{P_TN_s}}\mathbf{Z}_l\mathbf{s}^*$, it can be obtained that
{\small\begin{align}
U_l =& \left \|\sqrt{P_TN_s}\mathbf{h}_l + \bar{\mathbf{z}}_{l,N_s}\right \|_2^2 =  P_TN_s ||\mathbf{h}_l||_2^2 + ||\bar{\mathbf{z}}_{l,N_s}||_2^2 + \sqrt{P_TN_s} \mathbf{h}_l^\dagger \bar{\mathbf{z}}_{l,N_s} + \sqrt{P_TN_s} \bar{\mathbf{z}}_{l,N_s}^\dagger \mathbf{h}_l,\label{Eq:numerator}
\end{align}}
and
{\small\begin{align}
V_l = & P_TN_s ||\mathbf{h}_l||_2^2 + ||\mathbf{Z}_l||_F^2 + \sqrt{P_TN_s} \mathbf{h}_l^\dagger \bar{\mathbf{z}}_{l,N_s} + \sqrt{P_TN_s} \bar{\mathbf{z}}_{l,N_s}^\dagger \mathbf{h}_l-U_l =  ||\mathbf{Z}_l||_F^2 - ||\bar{\mathbf{z}}_{l,N_s}||_2^2.
\end{align}}

Construct a unitary matrix $\tilde{\mathbf{S}} \triangleq \left[\tilde{\mathbf{s}}_1,\ldots,\tilde{\mathbf{s}}_{N_s-1},\mathbf{s}^*/\sqrt{P_TN_s} \right] \in \mathbb{C}^{N_s\times N_s}$ with the last column being $\mathbf{s}^*/\sqrt{P_TN_s}$, then $V_l$ can further be represented as:
{\small\begin{equation}\label{Eq:denominator}
V_l = ||\mathbf{Z}_l\tilde{\mathbf{S}}||_F^2 - ||\bar{\mathbf{z}}_{l,N_s}||_2^2 = \sum_{n=1}^{N_s-1}||\bar{\mathbf{z}}_{l,n}||_2^2,
\end{equation}}
where $\bar{\mathbf{z}}_{l,n} = \mathbf{Z}_l\tilde{\mathbf{s}}_{n}$, $n=1,\ldots,N_s-1$.
It is easy to check that the vectors $\bar{\mathbf{z}}_{l,n}$, $n=1,\ldots,N_s$, are jointly Gaussian.
Since by construction, $\bar{\mathbf{z}}_{l,n}$, $n=1,\ldots,N_s$, are mutually uncorrelated, they are statistically independent.
As $U_l$ is only a function of $\bar{\mathbf{z}}_{l,N_s}$ and $V_l$ is a function of $\bar{\mathbf{z}}_{l,n}$, $n=1,\ldots,N_s-1$, $U_l$ and $V_l$ are independent.

As the noises are i.i.d., $U_l$s and $V_l$s are statistically independent.  It can therefore be concluded that $U\triangleq \sum_{l=1}^{L}U_l$ and $V\triangleq \sum_{l=1}^{L}V_l$ are independent.

From \eqref{Eq:numerator} and \eqref{Eq:denominator}, it is clear that under ${\cal H}_1$, $2U_l/\sigma^2$ has a non-central chi-square distribution with $2N_R$ degrees of freedom~(DoFs) and a non-centrality parameter $2N_sP_T\|\mathbf{h}_l\|_2^2/\sigma^2$, and $2V_l/\sigma^2$ admits a central chi-square distribution with $2N_R(N_s-1)$ DoFs.
Thus, $2U/\sigma^2$ has a non-central chi-square distribution with $2LN_R$ DoFs and a non-centrality parameter
{\small\[\lambda = 2P_TN_s/\sigma^2\sum_{l=1}^L\|\mathbf{h}_l\|_2^2\]}
 and $2V/\sigma^2$ admits a central chi-square distribution with $2LN_R(N_s-1)$ DoFs:
{\small\begin{equation}
\left\{\begin{array} {l}
        \frac{2U}{\sigma^2} \sim \chi^2_{2LN_R}(\lambda), \\
        \frac{2V}{\sigma^2} \sim \chi^2_{2LN_R(N_s-1)}.
        \end{array}
\right.
\end{equation}}
The proof of \eqref{F_distribution} under ${\cal H}_1$ is immediate by considering that $(N_s-1)L_G(\tau) = \frac{2U/(2LN_R\sigma^2)}{2V/(2LN_R(N_s-1)\sigma^2)}$.

The proof of \eqref{F_distribution} under hypothesis ${\cal H}_0$ is trivial by putting $\mathbf{h} = \mathbf{0}$ in the ${\cal H}_1$ case.
This concludes the proof of Proposition~\ref{Proposition_GLRT_sta}.

\section{Proof of Lemma~\ref{lemma_fading}}\label{proof_lemma}
The probability of miss-detection can be represented as $P_{miss} = Pr\{L_G(\tau)\leq\gamma|{\cal H}_1, \bar{h}_L<\underline{h}\}\times \xi +Pr\{L_G(\tau)\leq\gamma|{\cal H}_1, \bar{h}_L\geq\underline{h}\} \times (1-\xi)$.
Since $F(x|n_1,n_2,\lambda)$ is monotonically decreasing with respect to $\lambda$, provided that $n_1$ and $n_2$ are fixed~\cite{baharev2008computation}, it can be seen that $Pr\{L_G(\tau)\leq\gamma|{\cal H}_1, \bar{h}_L\geq\underline{h}\} \leq Pr\{L_G(\tau)\leq\gamma|{\cal H}_1, \bar{h}_L=\underline{h}\} = F(\gamma|2N_RL, 2NR_L(N_s-1),L\underline{\eta})$. It can then be obtained that for an arbitrary $\xi>0$ with $Pr\{\bar{h}_L<\underline{h}\}=\xi$,
{\small\begin{align}\label{Eq:lemma1_proof}
 P_{miss} & \leq \xi + (1-\xi)\times Pr\{L_G(\tau)\leq\gamma|{\cal H}_1, \bar{h}_L\geq\underline{h}\} \nonumber \\
 & \leq  \xi + (1-\xi)F(\gamma|2N_RL, 2NR_L(N_s-1),L\underline{\eta}).
\end{align}}

\section{Proof of Proposition~\ref{Proposition:LDP}} \label{proof_LDP}
Let $U$ and $V$ be the random variables defined in Appendix~\ref{App_pGLRT}, and denote $\tilde{U} \triangleq \frac{2U}{\sigma^2}$, $\tilde{V} \triangleq \frac{2V}{\sigma^2}$, $U_L \triangleq \frac{\tilde{U}}{L}$, $V_L \triangleq \frac{\tilde{V}}{L}$. Random variables $U_L$ and $V_L$ are independent as $U$ and $V$ are independent.

We first prove \eqref{LDP:L} by showing that $(U_L,V_L)$ satisfy the LDP, using the well known Gartner-Ellis Theorem. To demonstrate this, we need to show that the limiting logarithmic moment generation function~(MGF)
{\small \begin{equation} \label{lim_MGF}
\Lambda(\mathbf{t}) = \lim_{L\uparrow \infty} \frac{1}{L}\Lambda_{L}(L\mathbf{t})
\end{equation}}
exists as an extended real number~\cite[Chapter 2.3, pp. 43]{dembo2009large}, where $\Lambda_{L}(\mathbf{t}) \triangleq \log M_{(U_L,V_L)}(\mathbf{t})$ is the logarithmic MGF of $(U_L,V_L)$ and $\mathbf{t}=[t_1,t_2]$.

Due to the independence between $U_L$ and $V_L$, the MGF of $(U_L,V_L)$ is simply $M_{(U_L,V_L)}(\mathbf{t}) = M_{U_L}(t_1)M_{U_L}(t_2)$, where $M_{U_L}(t) =  \mathbb{E}_{U_L}\{e^{tU_L}\}$ and $M_{V_L}(t) =  \mathbb{E}_{V_L}\{e^{tV_L}\}$ are the MGFs for $U_L$ and $V_L$, respectively. We can therefore obtain that
{\small \begin{align}\label{log_MGF}
\Lambda_{L}(L\mathbf{t})  &= \log M_{U_L}(Lt_1)+ \log M_{V_L}(Lt_2) = \log \mathbb{E}_{\tilde{U}}\{e^{t_1\tilde{U}}\} + \log \mathbb{E}_{\tilde{V}}\{e^{t_2\tilde{V}}\}.
\end{align}}
Recall the facts shown in Appendix~\ref{App_pGLRT} that $\tilde{U}=\frac{2U}{\sigma^2} \sim \chi^2_{2LN_R}(\lambda)$ and $\tilde{V}=\frac{2V}{\sigma^2} \sim \chi^2_{2LN_R(N_s-1)}$, we can obtain $M_{\tilde{U}}(t) =  \mathbb{E}_{\tilde{U}}\{e^{t\tilde{U}}\}$ and $M_{\tilde{V}}(t) =  \mathbb{E}_{\tilde{V}}\{e^{t\tilde{V}}\}$ as follows:
{\small \begin{equation} \label{MFG_U}
\mathbb{E}_{\tilde{U}}\{e^{t_1\tilde{U}}\}  =  \left\{
									\begin{array}{ll}
										\frac{e^{\frac{\lambda t_1}{1-2t_1}}}{(1-2t_1)^{LN_R}}, & t_1<\frac{1}{2} \\
										+\infty, & \text{otherwise}
									   \end{array}
									   \right.
\end{equation}
\begin{equation} \label{MFG_V}
\mathbb{E}_{\tilde{V}}\{e^{t_2\tilde{V}}\}= \left\{
									\begin{array}{ll}
										(1-2t_2)^{-LN_R(N_s-1)}, & t_2<\frac{1}{2} \\
										+\infty. & \text{otherwise}
									   \end{array}
									   \right.
\end{equation}}

Using \eqref{MFG_U}, \eqref{MFG_V}, \eqref{log_MGF} and \eqref{lim_MGF}, it can be shown that
{\small \begin{align}
\small \Lambda(\mathbf{t}) &= \lim_{L\uparrow \infty} \frac{1}{L}\Lambda_{L}(L\mathbf{t}) = \left\{
		\begin{array}{ll}
			\frac{\eta t_1}{1-2t_1} - N_R\log(1-2t_1) -N_R(N_s-1)\log(1-2t_2), & t_1, t_2<\frac{1}{2} \\
			+\infty. & \text{otherwise}
		\end{array}
		 \right.
\end{align}}
Further, it can be shown that $\Lambda(\mathbf{t}) = 0$ only when $t_1<\frac{1}{2}$ and $t_2<\frac{1}{2}$. This verifies that $(U_L,V_L)$ satisfy the Gartner-Ellis conditions~\cite[Assumption 2.3.2, pp. 43]{dembo2009large}.

The rate function of $(U_L,V_L)$, i.e., $I_L(u,v)$ can then be obtained as
{\small \begin{align}
I_L(u,v)& \triangleq \sup_{t_1,t_2 \in \mathbb{R}} \{t_1u+t_2v - \Lambda(\mathbf{t})\} \nonumber \\
& = \sup_{t_1,t_2<\frac{1}{2}} \{t_1u+t_2v - \frac{\eta t_1}{1-2t_1} +  N_R\log(1-2t_1)+N_R(N_s-1)\log(1-2t_2) \} \label{rate_function_optimization} \\
& = t^*_1u+t^*_2v - \frac{\eta t^*_1}{1-2t^*_1} +  N_R\log(1-2t^*_1)+N_R(N_s-1)\log(1-2t^*_2), \label{Rate_function}
\end{align}}
where  $t^*_1 = \frac{1}{2}- \frac{N_R+\sqrt{N_R^2+\eta u}}{2u}$ and $t^*_2 = \frac{1}{2}-\frac{N_R(N_s-1)}{v}$.

With the rate function given in~\eqref{Rate_function}, the LDP tells us that
{\small \begin{equation}
\lim_{L\uparrow \infty} \frac{1}{L}\log Pr\{(U_L,V_L) \in {\cal A}\} = -\inf_{(u,v) \in {\cal A}}I_L(u,v),
\end{equation}}
if the set ${\cal A} \in \mathbb{R}^2$ is continuous.

Consider ${\cal A}=\{(u,v)|u \leq \gamma v\}$ as the collection of miss-detection events. This leads to the following rate function:
{\small \begin{equation} \label{optimisation_rate}
I^*(\eta,\gamma)= \inf_{u/v \leq \gamma}I_L(u,v).
\end{equation}}
Using the Karush-Kuhn-Tucker (KKT) conditions of \eqref{optimisation_rate},  it can be shown that:
\begin{enumerate}
\item If $u/v < \gamma$, then $I^*(\eta,\gamma) = I_L(u^*,v^*) = 0$, where $u^* = 2N_R+\eta$, $v^* = 2N_R(N_s-1)$. In this case, $\gamma>\frac{2N_R+\eta}{2N_R(N_s-1)}$.
\item If $u/v=\gamma$, then $I^*(\eta,\gamma) = I_L(u^*,v^*) = I_L(\gamma v^*,v^*)$,
where $v^*$ is obtained by solving $\frac{\partial I_L(\gamma v, v)}{\partial v} = 0$, or equivalently by solving the following equation
{\small \begin{equation}\label{eq_L_v}
\frac{\gamma+1}{2} - \frac{N_R+\sqrt{N_R^2+\eta \gamma v}}{2v} -\frac{N_R(N_s-1)}{v}=0.
\end{equation}}
Let $x\doteq \sqrt{N_R^2+ \eta \gamma v}$,  \eqref{eq_L_v} can be rewritten into:
{\small \begin{equation}\label{eq_L_x}
\frac{\gamma+1}{\eta\gamma}(x^2-N_R^2) - x - N_R -2N_R(N_s-1) = 0.
\end{equation}}
Denote $x^*>0$ is a solution of \eqref{eq_L_x}, then $v^*$ can be obtained as follows:
{\small \begin{equation}\label{v_star_L}
v^* = \frac{x^{*2}-N_R^2}{\eta\gamma}.
\end{equation}}
Substituting $u^*$, $v^*$ into \eqref{Rate_function}  yields \eqref{rate_L}.
\end{enumerate}
This has concluded the proof of  \eqref{LDP:L}.

%\section{Proof of Corollary~\ref{LDP_corollary}}\label{Proof_corollary}
We now prove the monotonicity of $I^*(\eta,\gamma)$ with respect to $\eta$ when $\gamma<\frac{2N_R+\eta}{2N_R(N_s-1)}$.
Towards this end, we note
{\small \begin{align}
\frac{\partial I^*(\eta,\gamma)}{\partial \eta} &= \frac{1}{2}\left(1- \frac{\gamma v^*}{N_R+\sqrt{N_R^2+\gamma \eta v^*}}\right) = \frac{1}{2}\frac{v^*-2N_R(N_s-1)}{N_R+\sqrt{N_R^2+\gamma \eta v^*}}, \label{derivative_I}
\end{align}}
where \eqref{derivative_I} follows from \eqref{eq_L_v}. To show $\frac{\partial I^*(\eta,\gamma)}{\partial \eta}>0$, it is sufficient to show
{\small \begin{equation}\label{condition:v}
v^*>2N_R(N_s-1).
\end{equation}}
Considering \eqref{v_star_L} and noticing the fact that $x^*>0$, showing \eqref{condition:v} is equivalent to show
{\small \begin{equation}\label{condition:x}
x^* > \sqrt{N_R^2+2\gamma\eta N_R(N_s-1)}.
\end{equation}}
Denote $f(x) \doteq \frac{\gamma+1}{\eta\gamma}(x^2-N_R^2) - x - N_R -2\gamma\eta N_R(N_s-1)$ as the left side of \eqref{eq_L_x}.
Noticing that $f(x^*)=0$ and $x^*>0$, it is sufficient to show $f(x=\sqrt{N_R^2+2\gamma\eta N_R(N_s-1)})=2\gamma N_R(N_s-1) - N_R - \sqrt{N_R^2+2\gamma\eta N_R(N_s-1)}<0$.

Since $\gamma<\frac{2N_R+\eta}{2N_R(N_s-1)}$, it can be obtained that
{\small \begin{align}
(2\gamma N_R(N_s-1) - N_R)^2 - \left(\sqrt{N_R^2+2\gamma\eta N_R(N_s-1)}\right)^2 = 2\gamma N_R(N_s-1)\left[2\gamma N_R(N_s-1)-(2N_R+\eta )\right] <0. \nonumber
\end{align}}
It is therefore concluded that $\frac{\partial I^*(\eta,\gamma)}{\partial \eta}>0$ when $\gamma<\frac{2N_R+\eta}{2N_R(N_s-1)}$.

\section{Proof of Proposition~\ref{proposition_optimal_beam}}\label{proof_proposition3}
The average miss-detection probability $\bar{p}_{miss}$
{\small \begin{align}
\bar{p}_{miss} = \int_{\Omega}p_{miss}(\phi)p(\phi)d\phi \geq \int_{\Omega^{-}}p_{miss}(\phi)p(\phi)d\phi & \geq \int_{\Omega^{-}}P_{miss}(\eta^{-},L,\gamma)p(\phi)d\phi \label{eq:miss_lower_bound}\\
& = P^-P_{miss}(\eta^{-},L,\gamma), \label{eq:miss_lower_bound_2}
\end{align}}
where $P^-\doteq\int_{\Omega^{-}}p(\phi)d\phi>0$ and \eqref{eq:miss_lower_bound} is due to the facts that $P_{miss}(\eta,L,\gamma)$ is monotonically decreasing with respect to $\eta$ when both $N_s$ and $L$ are fixed~\cite{baharev2008computation}, and $\eta(\phi)\leq \eta^{-}$, $\forall \phi \in \Omega^-$.

According to \eqref{eq:miss_lower_bound_2} and applying Proposition~\ref{Proposition:LDP}, it can be obtained that:
{\small \begin{align}
\lim_{L\uparrow \infty} -\frac{1}{L}\log\bar{p}_{miss}\leq \lim_{L\uparrow \infty} -\frac{1}{L}\left(\log P^- +\log P_{miss}(\eta^{-},L,\gamma)\right)= I^*(\eta^{-},\gamma),
\end{align}}
where $I^*(\eta,\gamma)$ is the rate function given by \eqref{rate_L}.
This concludes the proof.

\renewcommand{\baselinestretch}{1.15}


\begin{thebibliography}{10}
\providecommand{\url}[1]{#1}
\csname url@samestyle\endcsname
\providecommand{\newblock}{\relax}
\providecommand{\bibinfo}[2]{#2}
\providecommand{\BIBentrySTDinterwordspacing}{\spaceskip=0pt\relax}
\providecommand{\BIBentryALTinterwordstretchfactor}{4}
\providecommand{\BIBentryALTinterwordspacing}{\spaceskip=\fontdimen2\font plus
\BIBentryALTinterwordstretchfactor\fontdimen3\font minus
  \fontdimen4\font\relax}
\providecommand{\BIBforeignlanguage}[2]{{%
\expandafter\ifx\csname l@#1\endcsname\relax
\typeout{** WARNING: IEEEtran.bst: No hyphenation pattern has been}%
\typeout{** loaded for the language `#1'. Using the pattern for}%
\typeout{** the default language instead.}%
\else
\language=\csname l@#1\endcsname
\fi
#2}}
\providecommand{\BIBdecl}{\relax}
\BIBdecl
\bibitem{boccardi2014five}
F.~Boccardi, R.~W. Heath, A.~Lozano, T.~L. Marzetta, and P.~Popovski, ``Five
  disruptive technology directions for 5{G},'' \emph{IEEE Comms. Magazine},
  vol.~52, no.~2, pp. 74--80, 2014.

\bibitem{andrews2014will}
J.~G. Andrews, S.~Buzzi, W.~Choi, S.~V. Hanly, A.~Lozano, A.~C. Soong, and
  J.~C. Zhang, ``What will 5{G} be?'' \emph{IEEE J. Sel. Areas Comms.},
  vol.~32, no.~6, pp. 1065--1082, 2014.

\bibitem{7010531}
P.~Wang, Y.~Li, L.~Song, and B.~Vucetic, ``Multi-gigabit millimeter wave
  wireless communications for 5{G}: from fixed access to cellular networks,''
  \emph{IEEE Comms. Magazine}, vol.~53, no.~1, pp. 168--178, January 2015.

\bibitem{niu2015survey}
Y.~Niu, Y.~Li, D.~Jin, L.~Su, and A.~V. Vasilakos, ``A survey of millimeter
  wave communications (mmwave) for 5g: opportunities and challenges,''
  \emph{Wireless Networks}, vol.~21, no.~8, pp. 2657--2676, 2015.

\bibitem{roh2014millimeter}
W.~Roh, J.-Y. Seol, J.~Park, B.~Lee, J.~Lee, Y.~Kim, J.~Cho, K.~Cheun, and
  F.~Aryanfar, ``Millimeter-wave beamforming as an enabling technology for 5{G}
  cellular communications: theoretical feasibility and prototype results,''
  \emph{IEEE Comms. Magazine}, vol.~52, no.~2, pp. 106--113, 2014.

\bibitem{6717211}
O.~El~Ayach, S.~Rajagopal, S.~Abu-Surra, Z.~Pi, and R.~Heath, ``Spatially
  sparse precoding in millimeter wave mimo systems,'' \emph{IEEE Trans.
  Wireless Comms.}, vol.~13, no.~3, pp. 1499--1513, March 2014.

\bibitem{6600706}
S.~Hur, T.~Kim, D.~Love, J.~Krogmeier, T.~Thomas, and A.~Ghosh, ``Millimeter
  wave beamforming for wireless backhaul and access in small cell networks,''
  \emph{IEEE Trans. Comms.}, vol.~61, no.~10, pp. 4391--4403, October 2013.

\bibitem{el2013multimode}
O.~El~Ayach, R.~W. Heath, S.~Rajagopal, and Z.~Pi, ``Multimode precoding in
  millimeter wave {MIMO} transmitters with multiple antenna sub-arrays,'' in
  \emph{IEEE GLOBECOM}, 2013, pp. 3476--3480.

\bibitem{6834753}
M.~Akdeniz, Y.~Liu, M.~Samimi, S.~Sun, S.~Rangan, T.~Rappaport, and E.~Erkip,
  ``Millimeter wave channel modeling and cellular capacity evaluation,''
  \emph{IEEE J. Sel. Areas Comms.}, vol.~32, no.~6, pp. 1164--1179, June 2014.

\bibitem{li2013anchor}
Q.~C. Li, H.~Niu, G.~Wu, and R.~Hu, ``Anchor-booster based heterogeneous
  networks with mmwave capable booster cells,'' in \emph{IEEE Globecom
  Workshops}, 2013, pp. 93--98.

\bibitem{sesia2009lte}
S.~Sesia, I.~Toufik, and M.~Baker, \emph{LTE: the UMTS long term
  evolution}.\hskip 1em plus 0.5em minus 0.4em\relax Wiley Online Library,
  2009.

\bibitem{shokri2015millimeter}
H.~Shokri-Ghadikolaei, C.~Fischione, G.~Fodor, P.~Popovski, and M.~Zorzi,
  ``Millimeter wave cellular networks: A mac layer perspective,'' \emph{IEEE
  Trans. Comms}, vol.~63, no.~10, pp. 3437--3458, 2015.

\bibitem{abu2014synchronization}
S.~Abu-Surra, S.~Rajagopal, and X.~Zhang, ``Synchronization sequence design for
  mmwave cellular systems,'' in \emph{IEEE 11th Consumer Communications and
  Networking Conference (CCNC)}, 2014, pp. 617--622.

\bibitem{7094805}
V.~Desai, L.~Krzymien, P.~Sartori, W.~Xiao, A.~Soong, and A.~Alkhateeb,
  ``Initial beamforming for mmwave communications,'' in \emph{48th Asilomar
  Conference on Signals, Systems and Computers}, Nov 2014, pp. 1926--1930.

\bibitem{barati2014dreictional}
C.~Barati~Nt., S.~Hosseini, S.~Rangan, P.~Liu, T.~Korakis, S.~Panwar, and
  T.~Rappaport, ``Directional cell discovery in millimeter wave cellular
  networks,'' \emph{IEEE Trans. Wireless Comms.}, vol.~14, no.~12, pp.
  6664--6678, 2015.

\bibitem{dembo2009large}
A.~Dembo and O.~Zeitouni, \emph{Large deviations techniques and
  applications}.\hskip 1em plus 0.5em minus 0.4em\relax Springer Science \&
  Business Media, 2009, vol.~38.

\bibitem{alkhateeb2014channel}
A.~Alkhateeb, O.~El~Ayach, G.~Leus, and R.~W. Heath, ``Channel estimation and
  hybrid precoding for millimeter wave cellular systems,'' \emph{IEEE J. Sel.
  Sig. Processing}, vol.~8, no.~5, pp. 831--846, 2014.

\bibitem{bliss2010temporal}
D.~W. Bliss and P.~A. Parker, ``Temporal synchronization of mimo wireless
  communication in the presence of interference,'' \emph{IEEE Trans. Sig.
  Processing}, vol.~58, no.~3, pp. 1794--1806, 2010.

\bibitem{muhi2010modelling}
Z.~Muhi-Eldeen, L.~P. Ivrissimtzis, and M.~Al-Nuaimi, ``Modelling and
  measurements of millimetre wavelength propagation in urban environments,''
  \emph{IET microwaves, antennas \& propagation}, vol.~4, no.~9, pp.
  1300--1309, 2010.

\bibitem{baharev2008computation}
A.~Baharev and S.~Kem{\'e}ny, ``On the computation of the noncentral {F} and
  noncentral beta distribution,'' \emph{Statistics and Computing}, vol.~18,
  no.~3, pp. 333--340, 2008.

\bibitem{Xiao2016}
Z.~Xiao, T.~He, P.~Xia, and X.~G. Xia, ``Hierarchical codebook design for
  beamforming training in millimeter-wave communication,'' vol.~15, no.~5, pp.
  3380--3392, 2016.

\bibitem{raghavan2016beamforming}
V.~Raghavan, J.~Cezanne, S.~Subramanian, A.~Sampath, and O.~Koymen,
  ``Beamforming tradeoffs for initial ue discovery in millimeter-wave mimo
  systems,'' \emph{IEEE J. Sel. Sig. Processing}, vol.~10, no.~3, pp. 543--559,
  2016.

\bibitem{lebret1997antenna}
H.~Lebret and S.~Boyd, ``Antenna array pattern synthesis via convex
  optimization,'' \emph{IEEE Trans. Sig. Processing}, vol.~45, no.~3, pp.
  526--532, 1997.

\bibitem{goldberg2000genetic}
D.~E. Goldberg, \emph{Genetic Algorithms--In Search, Optimization \& Machine
  Learning}.\hskip 1em plus 0.5em minus 0.4em\relax Addison Wesley, New Delhi,
  1989.

\bibitem{huang2011millimeter}
K.-C. Huang and Z.~Wang, \emph{Millimeter wave communication systems}.\hskip
  1em plus 0.5em minus 0.4em\relax John Wiley \& Sons, 2011, vol.~29.
  
  \end{thebibliography}
\end{document}